%% file: Journal_arXiv_final.tex
\documentclass[journal]{IEEEtran}

\input{preamble.tex}

\title{  \vspace{0.0cm}  \centering  \Huge Nonlinear Dynamic Systems Parameterization Using Interval-Based Global Optimization: Computing Lipschitz Constants and Beyond}

\author{Sebastian A. Nugroh$\text{o}^{\star}$, \thanks{
		$^{\star}$Department of Electrical Engineering and Computer Science, University of Michigan, 1301 Beal Ave., Ann Arbor, MI 48109 (snugroho@umich.edu).} Ahmad F. Tah$\text{a}^{\dagger}$,\thanks{
		$^{\dagger}$Department of Civil and Environmental Engineering, Vanderbilt University, 2201 West End Ave., Nashville, TN 37235 (ahmad.taha@vanderbilt.edu).} and 
	Vu Hoan$\text{g}^{\ddagger}$ 
	\thanks{
		$^{\ddagger}$Department of Mathematics, The University of Texas at San Antonio, One UTSA Circle, San Antonio, TX 78249 (duynguyenvu.hoang@utsa.edu).}
	\thanks{
		This work is partially supported by Valero Energy Corporation and National Science Foundation (NSF) under Grant CMMI-1728629, CMMI-1917164, DMS-1614797, DMS-1810687, and ECCS-2151571, and CMMI-2152450.}
	\thanks{$^{\dagger}$Corresponding author.}
	\vspace{-0.4cm}

}

\begin{document}
	
\setlength{\abovedisplayskip}{3.5pt}
\setlength{\belowdisplayskip}{3.5pt}
\setlength{\abovedisplayshortskip}{3.1pt}
\setlength{\belowdisplayshortskip}{3.1pt}
		
\newdimen\origiwspc%
\newdimen\origiwstr%
\origiwspc=\fontdimen2\font
\origiwstr=\fontdimen3\font

\fontdimen2\font=0.64ex

\maketitle
\thispagestyle{plain}
\pagestyle{plain}

\begin{abstract}
Numerous state-feedback and observer designs for nonlinear dynamic systems (NDS) have been developed in the past three decades. These designs assume that NDS nonlinearities satisfy one of the following \textit{function set classifications}: bounded Jacobian, Lipschitz continuity, one-sided Lipschitz, quadratic inner-boundedness, and quadratic boundedness. These function sets are characterized by \textit{constant} scalars or matrices bounding the NDS' nonlinearities.
These constants \textit{(i)} depend on the NDS' operating region, topology, and parameters, and \textit{(ii)} are utilized to synthesize observer/controller gains.
Unfortunately, there is a near-complete absence of algorithms to compute such bounding constants. 
In this paper, we develop analytical then computational methods to compute such constants. First, for every function set classification, we derive analytical expressions for these bounding constants through global maximization formulations. Second, we utilize a derivative-free, interval-based global maximization algorithm based on branch-and-bound framework to numerically obtain the bounding constants. 
Third, we showcase the effectiveness of our approaches to compute the corresponding parameters on some NDS such as highway traffic networks and synchronous generator models. 
\end{abstract}
\vspace{-0.15cm}

\vspace{-0.35cm}

\section{Introduction}
\IEEEPARstart{C}{onsider} nonlinear dynamic systems (NDS) of the form
\begin{equation}\label{equ:ndsnew}
	\dot{\m x}(t) = \m A \m x(t) + \m G\m f(\m x, \m u), \; \; \m y(t) = \m C \m x(t),
\end{equation}
where in \eqref{equ:ndsnew}, $\m x$, $\m u$, $\m y$, and $\m f(\cdot)$ define the $n$ states, $m$ control inputs, $p$ output measurements, and system's nonlinearities while $\m A$, $\m G$, and $\m C$ are the corresponding system's matrices.   {The NDS model~\eqref{equ:ndsnew} represents a wide variety of systems such as transportation, water, and energy systems.} In the past three decades, numerous of Lyapunov-based methods have been developed to perform state/output-feedback control and robust state estimation for NDS~\eqref{equ:ndsnew}. {These methods often consider or assume that nonlinear dynamics in~\eqref{equ:ndsnew} belong to one of the following function sets of nonlinear functions: \textit{(a)} \textit{bounded Jacobian}, \textit{(b)} \textit{Lipschitz continuous}, \textit{(c)} \textit{one-sided Lipschitz} (OSL),  \textit{(d)} \textit{quadratic inner-boundedness} (QIB), and \textit{(e)} \textit{quadratically bounded} (QB).}
For example, $\m f(\cdot)$ is locally Lipschitz continuous in $\m \Omega$ if the following holds
\vspace{-0.1cm}
\begin{align*}
	\norm{\m f(\m x,\m u)-\m f(\hat{\m x},\m u)}_2 \leq \gamma_l \norm{\m x - \hat{\m x}}_2,
\end{align*}
for some $\gamma_{l} \geq 0$ and all  $(\m x, \m u), (\hat{\m x}, \m u)\in \mathbf{\Omega}$.
Tab.~\ref{tab:nonlinear_class} lists the mathematical definitions of the other function sets.

  {In this paper, we {investigate} NDS \textit{parameterization}. That is, computing \textit{constants} that bound these nonlinearities---we refer to these constants as \textit{NDS parameters} (or just simply \textit{parameters}).} For instance, the nonnegative constant $\gamma_l$ defines the Lipschitz continuity property in the above example. This parameterization is crucial as hundreds of observer/controller design are in fact relying on accurate values of these constants---see Tab.~\ref{tab:nonlinear_class}. 
In what follows, we perform a brief review on how these function set classifications have been utilized in some observer/controller design literature.  

In \cite{Phanomchoeng2010,ichalal2012observer,zemouche2006observer},  observer designs for Lipschitz NDS via solving {linear matrix inequalities} (LMIs) or semi-definite programs (SDPs) are proposed. As an alternative to tackle some limitations in Lipschitz-based observer, observers for OSL and QIB systems are developed in \cite{Abbaszadeh2010,benallouch2012observer,zhang2012full}. For bounded Jacobian systems, the corresponding observer designs are proposed in \cite{Phanomchoeng2010b,Phanomchoeng2011,Wang2018}. The authors in \cite{zemouche2013lmi} use these Jacobian bounds to propose a new approach of observer design for Lipschitz nonlinear systems, which is then improved in \cite{Jin2018} to increase scalability.
  {As for controller and observer-based stabilization designs, the authors in \cite{Siljak2000,Siljak2002} develop a feedback stabilization framework given for QB NDS nonlinearity.} Not to be confused with QIB, QB is a drastically different function set---see Tab.~\ref{tab:nonlinear_class} For Lipschitz NDS, observer-based control framework are proposed in \cite{Yadegar2018,zemouche2017robust,Ekramian2017}. Various robust $\mc{H}_{\infty}$ control designs assuming OSL and QIB conditions are developed in \cite{song2015robust,liu2014static,Rastegari2019,Gholami2019}. 

{ Appendix \ref{appdx:A} 
	lists detailed formulations of several LMI- and SDP-based formulations for observer/controller designs for NDS classes provided in Tab.~\ref{tab:nonlinear_class}. For example, considering that the nonlinear function $\m f(\cdot)$ in NDS \eqref{equ:ndsnew} is Lipschitz continuous
	with constant $\gamma_l\geq 0$ and given that $\m G = \m I$, solving the following LMI 
	$$ \begin{bmatrix}
		\m A ^{\top}\m P + \m P\m A -\m Y\m C - \m C ^{\top}\m Y ^{\top}  +\kappa\gamma_l^2 \m I& \m P^{\top}\\
		\m P & -\kappa \m I \end{bmatrix} \prec 0,$$
	for matrix variables $\m P\succ 0,\m Y$ and $\kappa>0$ yields an asymptotically stable estimation error through observer dynamics
	$$\dot{\hat{\m x}}(t) = \m A \hat{\m x}(t)+ {{\m f}}(\hat{\m x}, \m u)+\m L(\m y(t)  - \m C\hat{\m x}(t) ), $$ with gain $\m L = \m P^{-1}\m Y$.} Notice that the solution to the above LMI depends on the Lipschitz constant $\gamma_l$. 

From the aforementioned literature, we realize that most of these observer/controller designs are posed as SDPs that incorporate constants characterizing the nonlinearities in NDS. As such, finding the tightest bounds and constants is crucial.   {For example, in the case for Lipschitz functions, a larger $\gamma_l$ than the actual one can restrict the feasible set of the corresponding SDPs---this is commonly referred to as a \textit{conservative} Lipschitz constant.} On the contrary, if we approximate $\gamma_l$ such that the value is below the actual one, there is a possibility that, due to the dynamical behavior of the system, the Lipschitz condition characterized by $\gamma_l$ may be violated in certain points in $\m\Omega$.   {In turn, this can lead to an unstable closed-loop system.} Therefore, it is very important to obtain the tightest bounds that characterize $\m f(\cdot)$ for Lipschitz function sets and beyond.   {This importance is not only needed in the context of state estimation/stabilization of NDS, but also in other domain---for example, to provide a robustness certificate in neural network classifier \cite{szegedy2014intriguing}.}

We note the following: 
\textit{(i)} All of the aforementioned studies solely focus on developing methods for observer/controller designs while assuming that the NDS belongs to one of the aforementioned function sets, rather than attempting to find or design algorithms to obtain bounding constants.  
\textit{(ii)} The literature mostly focuses on small-scale dynamic systems where bounding constants can be obtained analytically. {In fact, we have found that although analytical computations of Lipschitz constants can be useful, such computations could return large Lipschitz constants that are way too conservative~\cite{Nugroho2019}.}  Therefore, in order to be able to use these observer designs for a high-dimensional large-scale NDS, there is a need for systematic methods that are dedicated for NDS parameterization.  

To the best of our knowledge, there are no computational methods in the literature that parameterize the nonlinearities in complex, large-scale NDS---beyond Lipschitz nonlinear systems.
With that in mind, some recent studies have proposed methods to parameterize some Lipschitz NDS only. The authors in \cite{wood1996estimation} propose a stochastic method to approximate Lipschitz constant, albeit only applicable for functions containing single variable. In \cite{paulavivcius2006analysis}, a method of computing Lipschitz constant for continuously differentiable real-valued function is presented. Recently, authors in \cite{SCHULZEDARUP2018135} propose fast numerical methods for computing Lipschitz constant using interval arithmetic and automatic differentiation. {  A new approach to find tight Lipschitz bounds based on SDPs for feed-forward neural networks is proposed in \cite{Fazlyab2019}, while authors in \cite{Chakrabarty2019} employ \textit{kernelized density estimation} to estimate Lipschitz constant using data generated from unknown dynamical systems. }
{ The computation of Lipschitz constant for synchronous generator models is investigated in \cite{Nugroho2019}, where the authors utilize \textit{low discrepancy sequence} (LDS) to under-approximate the Lipschitz constant and compare it with an analytical solution. 
	Realize that the majority of these works only focus on parameterizing one class of nonlinearity, that is, Lipschitz continuous. For the remaining function sets in Tab. \ref{tab:nonlinear_class}, no methodology has been put forth. 
	
	The objective of this work is to provide a framework for NDS parameterization for large-scale NDS via systematic methods to compute the corresponding constants/parameters for each class of aforementioned nonlinearities, considering arbitrary nonlinear models having continuous partial derivatives. These constants are computed by optimally solving global maximization problems over a box-shaped convex compact set through a global optimization algorithm based on \textit{interval arithmetic} (IA).} The paper contributions are summarized as follows: 

\setlength{\textfloatsep}{10pt}
\begin{table}[t]
	\centering
	\vspace{-0.1cm}
	\caption{Function sets of nonlinearity considered in this paper.}\label{tab:nonlinear_class}
	\vspace{-0.2cm}
	\renewcommand{\arraystretch}{1}
	\begin{tabular}{|m{0.70in}|m{2.4in}|}
		\hline
		\textbf{Class} $\vphantom{\left(\frac{v}{l}\right)}$ & \textbf{Mathematical Property} $\vphantom{\left(\frac{v}{l}\right)}$ \\[0.5ex] \hline\hline
		\textit{Bounded} \newline \textit{Jacobian} &$\vphantom{\left(\frac{v_f^{G^0}}{l}\right)} -\infty < \barbelow{f}_{ij} \leq \dfrac{\partial f_i}{\partial x_j}(\m x,\m u)\leq \bar{f}_{ij} < + \infty $,
		\newline $\vphantom{\left(\frac{v}{l_r}\right)} (\m x, \m u)\in \mathbf{\Omega}$, $\barbelow{f}_{ij}, \bar{f}_{ij} \in \mbb{R}$   \\ \hline
		\textit{Lipschitz} \newline \textit{Continuous} & $\vphantom{\left(\frac{v^i}{l}\right)}\norm{\m f(\m x,\m u)-\m f(\hat{\m x},\m u)}_2 \leq \gamma_l \norm{\m x - \hat{\m x}}_2$  \newline
		$(\m x, \m u), (\hat{\m x}, \m u)\in \mathbf{\Omega}$, $\gamma_l \in \mbb{R}_{+}$
		\\ \hline
		\textit{One-Sided} \newline \textit{Lipschitz} & $\vphantom{\left(\frac{v_f}{l}\right)} \langle\m G(\m f(\m x,\m u)-\m f(\hat{\m x},\m u)),\m x-\hat{\m x}\rangle\leq \gamma_s \norm{\m x - \hat{\m x}}_2^{2}$, \newline $(\m x, \m u), (\hat{\m x}, \m u)\in \mathbf{\Omega}$, $\gamma_s \in \mbb{R}$ \\ \hline
		\textit{Quadratically} \newline \textit{Inner-Bounded} & $\vphantom{\left(\frac{v_f}{l}\right)} \langle\m G(\m f(\m x,\m u)-\m f(\hat{\m x},\m u)),\m G(\m f(\m x,\m u)-\m f(\hat{\m x},\m u))\rangle \leq $\newline $\gamma_{q1} \norm{\m x - \hat{\m x}}_2^{2} +\gamma_{q2}\langle\m G(\m f(\m x,\m u)-\m f(\hat{\m x},\m u)),\m x-\hat{\m x}\rangle$,\newline $(\m x, \m u), (\hat{\m x}, \m u)\in \mathbf{\Omega}$, $\gamma_{q1},\gamma_{q2} \in \mbb{R}$ \\ \hline
		\textit{Quadratically} \newline \textit{Bounded} & $ \vphantom{\left(\frac{v_f}{l}\right)}\langle \m f(\m x),\m f(\m x)\rangle \leq  \m x^{\top}\m \Gamma^{\top}\m \Gamma \m x$,  
		\newline  $\m x\in \mathbfcal{X}$, $\m \Gamma\in\mbb{R}^{h\times n}$ \\
		\hline 
	\end{tabular}
	\vspace{-0.1cm}
\end{table}
\setlength{\floatsep}{10pt}

{ 
	\noindent $\bullet$ Proposing closed-form, analytical formulations presented as global maximization problems for computing constants that characterize bounded Jacobian, Lipschitz continuous, OSL, QIB, and QB function sets---see Tab. \ref{tab:nonlinear_class}. Our formulations allow NDS parameterization problems to be solved with global maximization algorithms. 
	It is demonstrated later that some of our formulations can be scalable for high-dimensional large-scale NDS since \textit{(a)} we only need to solve optimization problems for each unique form of nonlinearities, \textit{(b)} it is possible to consider only a small portion of optimization variables, and \textit{(c)} these problems can be solved via distributed computing technique.
	
	\noindent $\bullet$ 
	Utilizing a derivative-free interval-based global maximization algorithm based on branch-and-bound framework for performing NDS parameterization. The key advantage of this particular algorithm is that, since it does not necessitate the corresponding objective function to be differentiable, it can be used for performing parameterization for a wide class of NDS. Moreover, this algorithm provides an optimality guarantee in the sense that, when the algorithm terminates, it computes the \textit{best} upper and lower bounds for the objective function.   {This algorithm is adopted from \cite{moa2007interval} and herein we \textit{(a)} reformulate it for global maximization problems such that it can be directly used for NDS parameterization, \textit{(b)} present it using rigorous mathematical notations for the ease of analysis and implementation, and \textit{(c)} include an additional step that can aid the practical implementation of this algorithm. } 
	
	\noindent $\bullet$ Showcasing the effectiveness and applicability of the proposed framework for NDS parameterization as well as observer designs for different kinds of NDS including traffic networks, synchronous generator, and dynamics of a moving object. Our numerical test results suggest that the proposed approach provides Lipschitz constants that are considerably less conservative than the ones obtained from analytical computations.  
}

A preliminary version of this work appears in \cite{Nugroho2020insights}, where we \textit{(i)} discover a novel relation between Lipschitz and QIB function sets and \textit{(ii)} pinpoint mistakes made in the numerical section of \cite{Abbaszadeh2013,zhang2012full} and present the appropriate corrections.   {In addition, we also made all MATLAB codes used in our numerical experiments to be freely available online via GitHub \cite{Nugroho2020github}.} The structure of this paper is described as follows. 
Section \ref{sec:problem_formulation} showcases the problem formulation. In Section \ref{sec:nonlinear_classification}, closed-form formulations to determine the corresponding constants for all function sets through global maximization problems are all presented. Section \ref{sec:interval_optimization} presents the proposed algorithms for global optimization based on IA designed specifically for solving problems given in Section \ref{sec:nonlinear_classification}. Thorough numerical examples are provided in Section \ref{sec:numerical_tests} and the paper is concluded in Section \ref{sec:conclusion_future_work}. 

\vspace{0.0cm}
\noindent \textit{\textbf{Notation:}} Italicized, boldface upper and lower case characters represent matrices and column vectors: $a$ is a scalar, $\m a$ is a vector, and $\m A$ is a matrix. Matrix $\m I$ denotes the identity square matrix, whereas $\mO$ denotes a zero matrix of appropriate dimensions. The notations $\mathbb{N}$, $\mathbb{R}$, $\mathbb{R}_+$, and $\mathbb{R}_{++}$ denote the set of natural, real, non-negative, and positive real numbers. The notations $\mathbb{R}^n$ and $\mathbb{R}^{p\times q}$ denote the sets of row vectors with $n$ elements and matrices with size $p$-by-$q$ with elements in $\mathbb{R}$. {  The notation $\mathbb{S}^n$ denotes the set of $n$-dimensional symmetric matrices.} For any vector $\m x \in \mathbb{R}^{n}$, $\Vert\m x\Vert_2$ denotes the Euclidean norm of $\m x$, defined as $\Vert \m x\Vert_2 := \sqrt{\m x^{\top}\m x} $ , where $\m x^{\top}$ is the transpose of $\m x$. For $\m x, \m y\in\mbb{R}^n$, inner product is defined as $\langle\m x,\m y\rangle := \m x^{\top}\m y$. {  If $\m A$ is a matrix, $\norm {\mA}_2$ denotes its induced 2-norm, $\norm{\m A}_{{F}}$ denotes its Frobenius norm, and $A_{(i,j)}$ denotes its $i$-th and $j$-th element.} The set $\mathbb{I}(n)$ is defined as $\mathbb{I}(n) := \{i\in\mathbb{N}\,|\, 1\leq i \leq n\}$, which is usually used to represent the set of indices. The notions $\mathrm{D}_x$ and $\nabla_{\hspace{-0.05cm}x}$ denote the Jacobian matrix and gradient vector with respect to vector $\m x$. For a square matrix $\mA$,  $\lambda_{\mathrm{max}}(\mA)$ and $\lambda_{\mathrm{min}}(\mA)$ respectively return the maximum and minimum eigenvalues of $\mA$.   {The set of all intervals containing real numbers is denoted by $\mbb{IR}$, whereas an interval vector of dimension $n$ is denoted by $\mbb{IR}^n$.}


\vspace{-0.2cm}
\section{Problem Formulation}\label{sec:problem_formulation}
In this paper, we consider any NDS of the following form
\begin{align}\label{eq:gen_dynamic_systems}
	\dot{\m x}(t) &= \mA \m x (t) + \mG\m f(\m x, \m u) + \mB\m u(t),
\end{align}
\noindent where vectors $\m x\in \mathbfcal{X}\subset \mathbb{R}^n$ and $\m u\in \mathbfcal{U}\subset \mathbb{R}^m$ represent the state and input of the system, function $\m f :\mathbb{R}^n\times \mathbb{R}^m\rightarrow \mathbb{R}^g$ captures any nonlinear phenomena, and $\m A$, $\mB$, $\mG$ are known constant matrices of appropriate dimensions. For brevity, we drop the time dependence from now on such that $\m x := \m x(t)$; the same also applies to $\m u$. We define $\mathbf{\Omega}:= \mathbfcal{X}\times \mathbfcal{U}$ with $\mathbf{\Omega}\subset\mbb{R}^p$ so that the domain of $\m f(\cdot)$ can be simply referred to as $\mathbf{\Omega}$. The following assumptions are considered in this paper.
\vspace{-0.1cm}
\begin{asmp}\label{asmp:1}
	\reducewordspace  The conditions below apply to NDS \eqref{eq:gen_dynamic_systems}
	
	\begin{enumerate}
		\item { $\mathbf{\Omega}:= \prod_{i\in\mbb{I}(p)} {\Omega}_i$ is a  nonempty $p$-dimensional orthotope (may also be referred to as a \textit{hyperrectangle} or \textit{box}). That is, ${\Omega}_i = \left[\barbelow{\omega}_i, \bar{\omega}_i\right]$ where $\barbelow{\omega}_i, \bar{\omega}_i\in\mathbb{R}$ are known bounds.
			\item  $\m f(\cdot)\in \mathcal{C}^1$, i.e., it is continuously differentiable in $\mathbf{\Omega}$. }
	\end{enumerate}
	\vspace{-0.15cm}
\end{asmp}
\vspace{-0.1cm}
The first assumption dictates that each $x_i$ and $u_j$ belongs to convex compact sets $\mathcal{X}_i$ and $\mathcal{U}_j$ respectively. This assumption is not restrictive, seeing that cyber-physical systems have bounded states and inputs, where each upper and lower bounds define the operating regions.   {The second assumption ensures that the partial derivatives of $\m f(\cdot)$ exist and the gradient vector $\nabla_x f_i(\cdot)$ for all $i\in\mbb{I}(g)$ are bounded within the set $\mathbf{\Omega}$ \cite{trench2003introduction}.  }

  { 
	In this paper, we are concerned with obtaining a way to parameterize the nonlinear function $\m f(\cdot)$ for the function sets listed in Tab. \ref{tab:nonlinear_class}. This essentially boils down to computing the corresponding constants for each nonlinearity---this is referred to as \textit{parameterization}. Assumption \ref{asmp:1} implies that $\m f(\cdot)$ is Lipschitz continuous, and thus it also follows that $\m f(\cdot)$  belongs to the OSL and QIB function sets \cite{Abbaszadeh2010}. As the consequence, one can utilize a plethora of observer/controller designs for NDS of the form \eqref{eq:gen_dynamic_systems}.} 
{ It is worthwhile mentioning that, since $(\m x,\m u)\in\mathbf{\Omega}$, these function sets are valid in the region of interest $\mathbf{\Omega}$.}

To parameterize  $\m f(\cdot)$, first it is truly crucial to obtain well-defined formulations that allow the parameterization of these function sets to be performed effectively by employing optimization routines. The next section demonstrates how NDS parameterization can be posed as a global maximization problem over the set $\mathbf{\Omega}$.

\vspace{-0.2cm}
\section{Synthesis of Global Maximization Problems for NDS Parameterization}\label{sec:nonlinear_classification}
This section presents comprehensive formulations that pose the parameterization of $\m f(\cdot)$ as global maximization problems in which the objective functions are all given in closed-form expressions, thereby making them applicable for the majority of deterministic global optimization algorithms.

\vspace{-0.25cm}
\subsection{Bounded Jacobian}\label{ssec:bj}
Bounds for partial derivatives of $\m f(\cdot)$ denoted by $\barbelow{f}_{ij}$ and $\bar{f}_{ij}$ for each $i\in \mathbb{I}(g)$ and $j\in \mathbb{I}(n)$ can be computed as 
\begin{subequations}\label{eq:bj_all}
	\begin{align}
		\bar{f}_{ij} &= \max_{(\m x,\m u)\in \mathbf{\Omega}}\dfrac{\partial f_i}{\partial x_j}(\m x, \m u) \label{eq:bj_max} \\
		\barbelow{f}_{ij} &= \min_{(\m x,\m u)\in \mathbf{\Omega}}\dfrac{\partial f_i}{\partial x_j}(\m x, \m u) = -\max_{(\m x,\m u)\in \mathbf{\Omega}}-\dfrac{\partial f_i}{\partial x_j}(\m x, \m u)\label{eq:bj_min}.
	\end{align}
\end{subequations}
Note that there is a possibility for $\frac{\partial f_i}{\partial x_j}(\cdot)$ not to include each $x_j$ and $u_j$. It is possible to speed up the computation in solving \eqref{eq:bj_all} by reducing the search space $\mathbf{\Omega}$ to $\mathbf{\Omega}_{\m r_i}:= \mathbfcal{X}_{\m r_i}\times \mathbfcal{U}_{\m r'_i}$ where $\mathbfcal{X}_{\m r_i} = \prod_{j\in\mbb{I}(r_i)} \mathcal{X}_j$ and $\mathbfcal{U}_{\m r'_i} = \prod_{j\in\mbb{I}(r'_i)} \mathcal{U}_j$ such that every $x_j\in \mathcal{X}_j$ and $u_j\in \mathcal{U}_j$ are variables included in $\frac{\partial f_i}{\partial x_j}(\cdot)$. 
\vspace{-0.1cm}
\begin{myrem}\label{rem:search_space_reduction}
	\reducewordspace
	The reduction of $\mathbf{\Omega}$ can also be applied for parameterizing other function sets whenever the objective function to be maximized does not include all variables from $\m x$ and $\m u$. This potentially allows faster computational time as well as reduces the memory size required to perform the computation.
\end{myrem}
\vspace{-0.1cm}

\vspace{-0.25cm}
\subsection{Lipschitz Continuous \vspace{-0.05cm}}\label{ssec:lc} 
Since it is assumed that Assumption \ref{asmp:1} holds for NDS \eqref{eq:gen_dynamic_systems}, $\m f(\cdot)$ always satisfies the Lipschitz continuity condition. Thus, the problem here is  determining the tightest Lipschitz constant for $\m f(\cdot)$ that is still useful for observer/controller design. According to \cite{marquez2003nonlinear,khalil2002nonlinear}, Lipschitz constant $\gamma_l$ for such NDS is nothing but the maximum value of the norm of Jacobian matrix of $\m f(\cdot)$ over the set $\mathbf{\Omega}$, that is 
\begin{align}
	\vspace{-0.2cm}
	\gamma_l = \max_{(\m x,\m u)\in\mathbf{\Omega}} \norm{\mathrm{D}_x \m f(\m x, \m u)}_2.\label{eq:lip_jacobian}
	\vspace{-0.2cm}
\end{align}
Although \eqref{eq:lip_jacobian} provides a straightforward method to compute $\gamma_l$, it is unfortunately difficult to obtain a closed-form expression of \eqref{eq:lip_jacobian} due to the use of {the 2-norm operator for the Jacobian matrix. Since we use euclidean norms $\norm{\cdot}_2$ on $\mathbb{R}^n$, the value of the operator norm (see \cite{marquez2003nonlinear}) is given by 
	\begin{align}
		\norm{\mathrm{D}_x \m f(\m x, \m u)}_2 = \sqrt{\lambda_{\operatorname{max}}\left(\mathrm{D}_x \m f(\m x, \m u) \,\mathrm{D}_x \m f(\m x, \m u)^\top\right)}. \nonumber
\end{align}}
We note that having a closed-form formulation for $\gamma_l$ is important as many deterministic global optimization methods necessitate the closed-form expression of the objective function. 
With that in mind, the following provides a closed-form formulation to compute $\gamma_l$, which is given in Theorem \ref{prs:Lipschitz_less_cons}---the following \textit{mean value theorem} and \textit{key lemma} are presented beforehand due to its significant role in constructing this theorem. 
\vspace{-0.15cm}
\begin{mylem}[\textit{Mean Value Theorem}]\label{lem:MVT}
	\reducewordspace
	Let the function $f:\mbb{R}\rightarrow \mbb{R}$ be continuous on $\left[a,b\right]$ and differentiable on $\left(a,b\right)$ for $a,b\in\mbb{R}$. Then there exists $\bar{\tau} \in \left(a,b\right)$ such that
	\begin{align}
		\frac{d f}{d x}(\bar{\tau}) = \frac{f(b)-f(a)}{b-a}.\nonumber
	\end{align}
\end{mylem}
\vspace{-0.15cm}
The proof of Lemma \ref{lem:MVT} is available in \cite{trench2003introduction}.   {We emphasize that Lemma \ref{lem:MVT} holds for scalar functions $f$ only, and not for vector-valued $\m f$ having values in $\mathbb{R}^m$.} In particular, for a
vector valued function $\m f : [a, b] \to \mathbb{R}^m$, it is not true that there exists a $a < \bar\tau < b$ such that 
\begin{align}
	\m f(b)-\m f(a) = \m f'(\bar\tau) (b-a)
\end{align}
holds. This is shown on p.112 of \cite{Rudin1976PM}.   {However, it is true that a $a < \bar\tau < b$ exists with} 
\begin{align}
	\norm{\m f(b)-\m f(a)}_2\leq \norm{\m f'(\bar \tau)}_2 |b-a|.
\end{align}
(see Theorem 5.19 in \cite{Rudin1976PM}). In deducing bounds for Lipschitz constants we will therefore proceed carefully and prove the next Lemma first, which we call key lemma in view of its importance. It will allow us to give a uniform and streamlined derivation of our bounds for Lipschitz constants, OSL constants and constants for the QIB condition. We note that the idea of the key lemma (taking the scalar product of $\m f(b)-\m f(a)$ with an auxiliary vector $\m w$) is inspired by \cite{Rudin1976PM} (see also \cite{khalil2002nonlinear}, Lemma 3.1). A first consequence of the key lemma is the explicit estimate \eqref{eq:locally_Lipschitz_theorem} for the Lipschitz constant. An alternative proof can be given by combining \eqref{eq:lip_jacobian} with the following inequality 
\begin{align*}
	\norm{\mathrm{D}_x\m f(\m x, \m u)}_2 \leq \norm{\mathrm{D}_x\m f(\m x, \m u)}_{{F}},
\end{align*}
While therefore it might be argued that \eqref{eq:locally_Lipschitz_theorem} is a direct consequence of well-known results, the key lemma is applicable in more general situations, and can be used to derive explicit bounds for OSL (see Section \ref{ssec:osl}) and QIB constants (see Section \ref{ssec:qib}). These results are, to the best of our knowledge, fully novel.
\vspace{-0.15cm}
\begin{mylem}[{\textit{Key Lemma}}]\label{lem:key}
	\reducewordspace Let $\m F\in \mathbb{R}^{g\times g}$ be an arbitrary matrix and $\m{w}\in \mathbb{R}^g$ and $\m x, \hat{\m x}\in \mathbb{R}^n, \m u \in \mathbf{R}^m$ such that $(\m x, \m u)\in\mathbf{\Omega}$. Then there exists a $\bar \tau\in (0, 1)$ such that
	\begin{subequations}
		\begin{align}\label{eq:key1}
			\begin{split}
				\langle \m F(\m f(\m x, \m u) - &\m f(\hat{\m x}, \m u) ), \m{w} \rangle =\\ &\quad\langle \mathbf{\Xi}(\hat{\m x} + \bar \tau (\m x -\hat{\m x} ), \m u)(\m x -\hat{\m x} ), \m{w}\rangle. 
			\end{split}
		\end{align}
		where $\m \Xi:=\m \Xi(\m x, \m u)\in \mathbb{R}^{g\times n}$ is given by
		\begin{align}
			{\m \Xi}_{(i,j)} (\m x,\m u) := \sum_{k\in \mbb{I}(g)}\hspace{-0.1cm} F_{(i,k)}\dfrac{\partial f_k}{\partial x_j}(\m x,\m u).\label{eq:one_sided_Lipschitz_matrix1}
		\end{align}
	\end{subequations}
\end{mylem}
\begin{proof}
	Let $(\m x, \m u), (\hat{\m x}, \m u)\in \mathbf{\Omega}$. Define two functions $\phi:\mathbb{R}\rightarrow \mathbb{R}$ and $\m g:\mathbb{R}\rightarrow \mathbb{R}^n$ as follows
	\begin{subequations}
		\begin{align*}
			\phi(\tau) &:= \langle \m F \m f(\m g(\tau)), \m w\rangle,\;\;
			\m g(\tau) := \hat{\m x} + \tau (\m x -\hat{\m x} ),
		\end{align*}
		Note that due to the assumed convexity of $\m \Omega$, $(g(\tau), \m u)$ is contained in $\m \Omega$ for all $\tau\in [0, 1]$. Since $\phi(\cdot)$ is continuous on $[0,1]$ and differentiable on $(0,1)$, then according to mean value theorem there exists $\bar{\tau}\in(0,1)$ such that
		\begin{align}
			\frac{d\phi}{d\tau}(\bar{\tau}) &= \frac{\phi(1)-\phi(0)}{1-0} = \phi(1)-\phi(0).
			\label{eq:key_proof1}
		\end{align}
		The left-hand side of \eqref{eq:key_proof1} reads 
		\begin{align*}
			\sum_{i\in \mbb{I}(g)}\sum_{k\in \mbb{I}(g)}\hspace{-0.1cm} F_{(i,k)}\dfrac{\partial f_k}{\partial x_j}(\m g(\bar \tau),\m u) (x_j - \hat x_j) w_i, 
		\end{align*}
		yielding the conclusion of our lemma.
	\end{subequations}
\end{proof}
\vspace{-0.15cm}
{As a first application of the key lemma, in what follows we obtain an expression for the Lipschitz constant.}
\vspace{-0.15cm}
\begin{theorem}\label{prs:Lipschitz_less_cons}
	\reducewordspace
	The nonlinear function $\m f :\mathbb{R}^n\times \mathbb{R}^m\rightarrow \mathbb{R}^g$ in \eqref{eq:gen_dynamic_systems} is locally Lipschitz continuous in $\mathbf{\Omega}$ such that for any $(\m x, \m u), (\hat{\m x}, \m u)\in \mathbf{\Omega}$ the following condition holds
	\begin{subequations}\label{eq:locally_Lipschitz_theorem}
		\vspace{-0.1cm}
		\begin{align}
			\norm{\m f(\m x,\m u)-\m f(\hat{\m x},\m u)}_2 \leq \gamma_{l_1} \norm{\m x - \hat{\m x}}_2, \label{eq:locally_Lipschitz}
		\end{align}
		where $\gamma_{l_1}$ can be obtained from solving the following problem
		\begin{align}
			\gamma_{l_1} =\Bigg(\max_{(\m{x},\m u)\in \mathbf{\Omega}}\sum_{i\in \mbb{I}(g)}\norm{\nabla_{\hspace{-0.05cm}x} f_i(\m{x},\m u)}_2^2\Bigg)^{\hspace{-0.05cm}\!1/2}. \label{eq:gamma_Lipschitz}
		\end{align}
	\end{subequations}
\end{theorem}

\begin{proof}
	Applying Lemma \ref{lem:key} with $\m F= \m I$ and $\m w := \m f(\m x, \m u) - \m f(\hat{\m x}, \m u)$ yields
	\begin{align*}
		{\m \Xi}_{(i,j)} (\m x,\m u) := \dfrac{\partial f_i}{\partial x_j}(\m x,\m u),
	\end{align*}
	and therefore the right-hand side of \eqref{eq:key1} reads
	\begin{align*}
		\Psi := \sum_{i\in \mbb{I}(g)} \sum_{j\in \mbb{I}(n)} \dfrac{\partial f_i}{\partial x_j}(\hat{\m x}+\bar \tau (\m x -\hat {\m x}),\m u) (x_j-\hat x_j) w_i.
	\end{align*}
	We now proceed to estimate $\Psi$.   {By applying the Cauchy-Schwarz inequality twice (first to the sum over $j$ and then over $i$), one can verify that}
	\begin{align*}
		&|\Psi|^2\leq\Bigg(\sum_{i\in \mbb{I}(g)} \norm{\nabla_{\m x} f_i(\hat{\m x}+\bar \tau (\m x -\hat {\m x}),\m u)}_{2} \norm{\m x -\hat {\m x}}_2 w_i\Bigg)^{\hspace{-0.1cm}2}\\
		&\leq \Bigg(\sum_{i\in \mbb{I}(g)} \norm{\nabla_{\m x} f_i(\hat{\m x}+\bar \tau (\m x -\hat {\m x}),\m u)}_{2}^2\Bigg)\hspace{-0.05cm}\norm{\m x -\hat {\m x}}_2^2~ \norm{\m w}_2^2\\
		&\leq \Bigg(\max_{(\m{x},\m u)\in \mathbf{\Omega}} \sum_{i\in \mbb{I}(g)} \norm{\nabla_{\m x} f_i(\m x,\m u)}_{2}^2\Bigg)\hspace{-0.05cm}\norm{\m x -\hat {\m x}}_2^2\norm{\m w}_2^2.
	\end{align*}
	Revisiting \eqref{eq:key1}, we see that we have just shown
	\begin{align*}
		|\langle  f(\m x, \m u) - \m f(\hat{\m x}, \m u) , \m w \rangle|^2 \leq \gamma_{l_1}^2 \norm{\m x -\hat {\m x}}_2^2~ \norm{\m w}_2^2
	\end{align*}
	Using our definition of $\m w=\m f(\m x, \m u) - \m f(\hat{\m x}, \m u)$ this reduces to 
	\begin{align*}
		\norm{f(\m x, \m u) - \m f(\hat{\m x}, \m u)}_2^4 \leq \gamma_{l_1}^2 \norm{\m x -\hat {\m x}}^2 \norm{\m f(\m x, \m u) - \m f(\hat{\m x}, \m u)}_2^2
	\end{align*}
	which immediately implies \eqref{eq:locally_Lipschitz} provided that $\norm{\m w}_2^2>0$. In the case that $\norm{\m w}_2^2=0$, \eqref{eq:locally_Lipschitz} holds trivially, so this completes the proof of the theorem.
\end{proof}
\vspace{-0.15cm}
Since the closed-form expression of \eqref{eq:gamma_Lipschitz} can be determined, unlike the one provided in \eqref{eq:lip_jacobian}, Lipschitz constant $\gamma_{l_1}$ can be computed using any deterministic global optimization algorithms. Notice that when the domain of interest $\mathbf{\Omega}$ is large and the function $\m f(\cdot)$ is of high dimension, solving \eqref{eq:gamma_Lipschitz} may take a significant amount of computational time. In this regard, it is possible to `split' the maximization problem \eqref{eq:gamma_Lipschitz} into several smaller sub-problems so that it may be solved in a distributed fashion, as described in the following corollary.
\vspace{-0.15cm} 
\begin{mycor}\label{cor:Lipschitz}
	\reducewordspace
	The nonlinear function $\m f :\mathbb{R}^n\times \mathbb{R}^m\rightarrow \mathbb{R}^g$ in \eqref{eq:gen_dynamic_systems} is locally Lipschitz continuous in $\mathbf{\Omega}$ satisfying \eqref{eq:locally_Lipschitz} with Lipschitz constant
	\begin{align}
		\gamma_{l_2} =\Bigg(\sum_{i\in \mbb{I}(g)}\max_{(\m{x},\m u)\in \mathbf{\Omega}}\norm{\nabla_{\hspace{-0.05cm}x} f_i(\m{x},\m u)}_2^2\Bigg)^{\hspace{-0.05cm}\!1/2}. \label{eq:gamma_Lipschitz_2}
	\end{align} 	
\end{mycor}
\begin{proof}
	For all $(\m x^*, \m u^*)\in \mathbf{\Omega}$ and all $i\in \mbb{I}(g)$, 
	\begin{align*}
		\norm{\nabla_{\hspace{-0.05cm}x} f_i(\m{x}^*,\m u^*)}_2^2\leq \max_{(\m x, \m u)\in \mathbf{\Omega}} \norm{\nabla_{\hspace{-0.05cm}x} f_i(\m{x},\m u)}_2^2 
	\end{align*}
	and hence
	\begin{align}\label{eq:cor_1_proof_2}
		\sum_{i\in \mbb{I}(g)}\norm{\nabla_{\hspace{-0.05cm}x} f_i(\m{x}^*,\m u^*)}_2^2 \leq \sum_{i\in \mbb{I}(g)} \max_{(\m x, \m u)\in \mathbf{\Omega}} \norm{\nabla_{\hspace{-0.05cm}x} f_i(\m{x},\m u)}_2^2 
	\end{align}
	  {This shows that $\gamma_{l_1} \leq \gamma_{l_2}$ and hence $\gamma_{l_2}$ in \eqref{eq:gamma_Lipschitz_2} is also a Lipschitz constant satisfying \eqref{eq:locally_Lipschitz}.}
\end{proof}
\vspace{-0.15cm}
{From $\gamma_{l_1} \leq \gamma_{l_2}$ we see that using $\gamma_{l_2}$ as a Lipschitz constant} is potentially more conservative than $\gamma_{l_1}$ albeit \eqref{eq:gamma_Lipschitz_2} can be solved individually for each $f_i(\cdot)$ to reduce computational time. This introduces trade-offs between \eqref{eq:gamma_Lipschitz} and \eqref{eq:gamma_Lipschitz_2}.  

\vspace{-0.1cm}
\subsection{One-Sided Lipschitz (OSL)}\label{ssec:osl}
OSL is another function set that is a generalization of Lipschitz continuity: a Lipschitz continuous function is also OSL \cite{Abbaszadeh2010}.
Realize that OSL constant $\gamma_s$ can be any real number while Lipschitz constant can only be nonnegative, as described in Tab. \ref{tab:nonlinear_class}.   {The following example illustrates this point.}
\vspace{-0.2cm}
\begin{exmpl}\label{exmpl:osl}
	Consider the function $f:\mathbb{R} \rightarrow \mathbb{R}$ defined as
	\begin{align}
		~~f(x) = - x^3 - 100 x.  \label{eq:example}
	\end{align}
	This function	is globally OSL with constant $\gamma_s = -100$. This can be seen from expanding the left-hand side of the OSL condition described in Tab. \ref{tab:nonlinear_class}
	\begin{align*}
		(f(x) - f(\hat x))(x-\hat x) &= (-x^3 + \hat{x}^3 - 100 (x-\hat x))(x-\hat x)\\
		&\leq -100 (x-\hat x)^2,
	\end{align*}
	which is due to $(-x^3 + \hat{x}^3)(x-\hat x)\leq 0$. On the other hand, an estimation of Lipschitz constant for $f(\cdot)$, say on the interval $\mathbfcal{X} = [-1, 1]$ would be at least equal to $\gamma_l = 103$. Thus, the sign information $\gamma_s < 0$ might give an advantage over $\gamma_l$ in their application for observer design, as argued in \cite{Abbaszadeh2010}.   {Fig. \ref{Fig:LIP_OSL} provides an illustration on function $f(\cdot)$ described in \eqref{eq:example} with Lipschitz and OSL conditions.} 
\end{exmpl} 
\vspace{-0.1cm}
In this section, we derive numerical methods to compute OSL constant $\gamma_s$ for NDS of the form \eqref{eq:gen_dynamic_systems} so that it is possible to have $\gamma_s\in\mbb{R}$. In this regard, first we propose numerical formulations that provide lower and upper bounds towards the left-hand side of OSL condition presented in Tab. \ref{tab:nonlinear_class}, which is summarized in the following proposition.

\begin{figure}[t]
	\centering
	\vspace{-0.5cm}
	\hspace{-0.2cm}
	\subfloat[]{\includegraphics[scale=0.175]{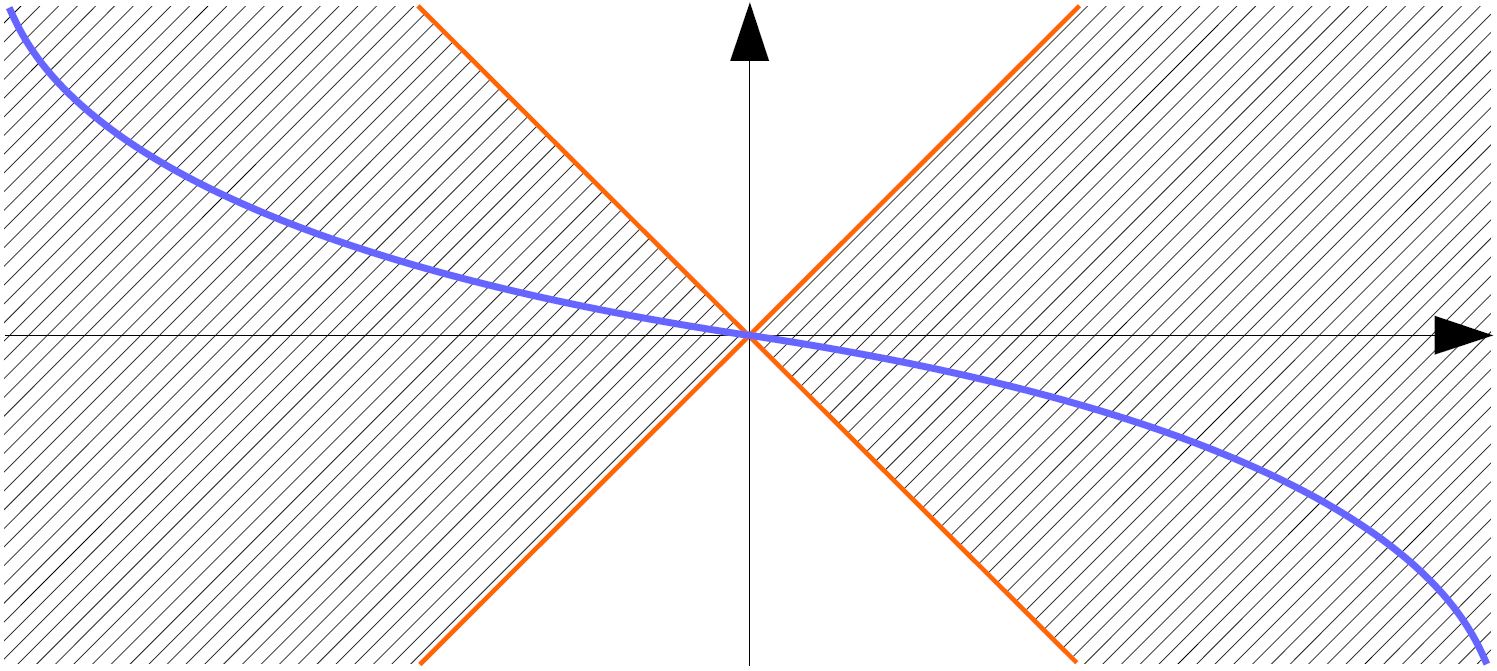}\label{Fig:LIP}}{\vspace{-0.2cm}}\hspace{-0.1em}
	\subfloat[]{\includegraphics[scale=0.175]{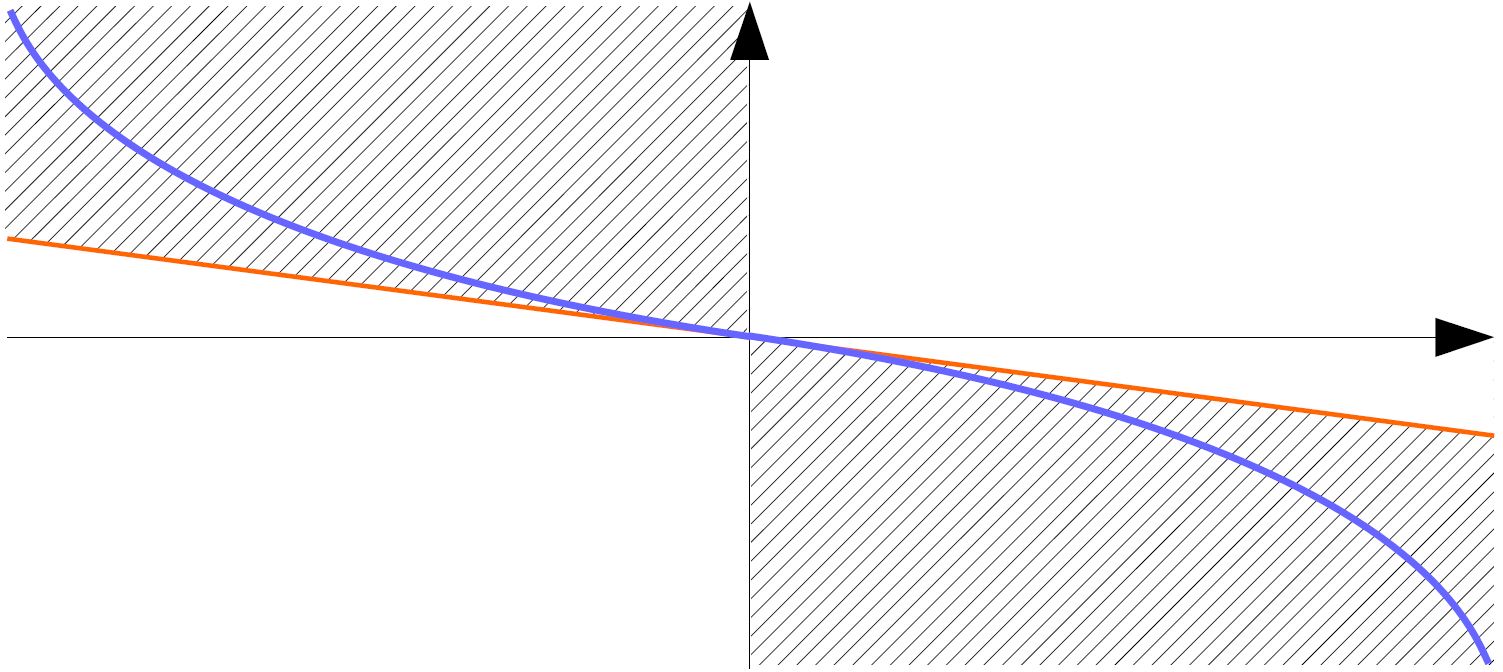}\label{Fig:OSL}}{}\hspace{-0.1em}
	\vspace{-0.1cm}
	\caption{Since the function $f(x) = - x^3 - 100 x$ (in blue) lies in the shaded area, then it satisfies two conditions: \textit{(a)} Lipschitz continuous and \textit{(b)} OSL. Notice that the OSL condition is characterized by a single straight line (in orange) having negative slope. }
	\label{Fig:LIP_OSL}	
	\vspace{-0.2cm}
\end{figure}

\vspace{-0.15cm}
\begin{myprs}\label{prs:osl_bounds}
	\reducewordspace
	For the nonlinear function $\m f :\mathbb{R}^n\times \mathbb{R}^m\rightarrow \mathbb{R}^g$ in \eqref{eq:gen_dynamic_systems}, there exist $\bar{\gamma},\barbelow{\gamma}\in\mathbb{R}$ such that for any $(\m x, \m u), (\hat{\m x}, \m u)\in \mathbf{\Omega}$ the following condition holds
	\begin{subequations}\label{eq:one_sided_Lipschitz_bounds_theorem}
		\begin{align}
			\barbelow{\gamma} \norm{\m x - \hat{\m x}}_2^{2}\leq \langle\m G(\m f(\m x,\m u)-\m f(\hat{\m x},\m u)),\m x-\hat{\m x}\rangle\leq \bar{\gamma} \norm{\m x - \hat{\m x}}_2^{2}, \label{eq:one_sided_Lipschitz}
		\end{align}
		where $\bar{\gamma}$ and $\barbelow{\gamma}$ are given as
		\begin{align}
			\bar{\gamma} &= \max_{(\m{x},\m u)\in \mathbf{\Omega}}\lambda_{\mathrm{max}}\left(\frac{1}{2}\left(\mathbf{\Xi} (\m x,\m u) + \mathbf{\Xi}^{\top} (\m x,\m u)\right)\right) \label{eq:one_sided_Lipschitz_upper}\\
			\barbelow{\gamma} &= \min_{(\m{x},\m u)\in \mathbf{\Omega}}\lambda_{\mathrm{min}}\left(\frac{1}{2}\left(\mathbf{\Xi}(\m x,\m u) + \mathbf{\Xi}^{\top} (\m x,\m u)\right)\right),\label{eq:one_sided_Lipschitz_lower}
		\end{align}
		where each of the $i$-th and $j$-th element of $\mathbf{\Xi} (\cdot)$ is specified as
		\begin{align}
			{\Xi}_{(i,j)} (\m x,\m u) := \sum_{k\in \mbb{I}(g)}\hspace{-0.1cm}G_{(i,k)}\dfrac{\partial f_k}{\partial x_j}(\m x,\m u).\label{eq:one_sided_Lipschitz_matrix}
		\end{align}
	\end{subequations}
\end{myprs}
\begin{proof}
	Applying key lemma with $\m w := \m x - \hat{\m x}$, we get
	\begin{align}
		&\langle \m G( f(\m x, \m u) - \m f(\hat{\m x}, \m u) ), (\m x -\hat{\m x}) \rangle \nonumber \\
		&= (\m x -\hat{\m x} )^{\top}\m\Xi^{\top}(\m z,\m u)(\m x -\hat{\m x} )\label{eq:osl_bounds_proof_5}
	\end{align}
	where the matrix $\m\Xi(\cdot)$ in \eqref{eq:osl_bounds_proof_5} is described in \eqref{eq:one_sided_Lipschitz_matrix} and $\m z =\hat{\m x} + \bar \tau(\m x - \hat{\m x})$ for some $\bar \tau\in [0, 1]$.
	Now \eqref{eq:osl_bounds_proof_5} is equal to
	\begin{align*}
		(\m x -\hat{\m x} )^{\top}\left(\frac{1}{2}\left(\mathbf{\Xi} (\m z,\m u) + \mathbf{\Xi}^{\top} (\m z,\m u)\right)\right)(\m x -\hat{\m x} ).
	\end{align*} Since $\left(\frac{1}{2}\left(\mathbf{\Xi} (\m z,\m u) + \mathbf{\Xi}^{\top} (\m z,\m u)\right)\right)$ is symmetric, we can apply the Rayleigh quotient \cite{horn2013matrix} to obtain{ 	
		\begin{align*}\barbelow{\gamma} \norm{\m x - \hat{\m x}}_2^{2} \leq \langle\m G(\m f(\m x,\m u)-\m f(\hat{\m x},\m u)),\m x-\hat{\m x}\rangle &\leq \bar{\gamma} \norm{\m x - \hat{\m x}}_2^{2}, 
	\end{align*}}
	  {which establishes \eqref{eq:one_sided_Lipschitz} whilst $\bar{\gamma}$ and $\barbelow{\gamma}$ are  given by \eqref{eq:one_sided_Lipschitz_upper} and \eqref{eq:one_sided_Lipschitz_lower} respectively.}
\end{proof}
\vspace{-0.15cm}
From Proposition \ref{prs:osl_bounds}, the OSL constant for NDS \eqref{eq:gen_dynamic_systems} is given by $\gamma_s = \bar{\gamma}$. This result generalizes the approach to compute OSL constant in \cite{Abbaszadeh2010} in two ways. First, our result applies for a more general form of NDS expressed in \eqref{eq:gen_dynamic_systems} and secondly, we also obtain a lower bound for the left-hand side of OSL condition presented in Tab. \ref{tab:nonlinear_class}, which is useful for determining QIB constants, as later explained in the next section. Yet, the non closed-form expression for $\bar{\gamma}$ described in \eqref{eq:one_sided_Lipschitz_upper} and similarly in \cite{Abbaszadeh2010} makes it difficult to compute via many deterministic global optimization methods. { To that end, instead of solving  \eqref{eq:one_sided_Lipschitz_upper}, it is  more desirable to solve another problem that, if successfully solved, produces an upper bound towards the solution of \eqref{eq:one_sided_Lipschitz_upper}.}
In what follows we discuss several solutions to this problem. First, a more direct conservative way to compute $\gamma_s$ is proposed below. 
\vspace{-0.15cm}
\begin{theorem}\label{cor:osl_bounds}
	\reducewordspace
	The nonlinear function $\m f :\mathbb{R}^n\times \mathbb{R}^m\rightarrow \mathbb{R}^g$ in \eqref{eq:gen_dynamic_systems} is OSL continuous in $\mathbf{\Omega}$ satisfying
	\begin{subequations} 
		\begin{align}
			\vphantom{\left(\frac{v_f}{l}\right)} \langle\m G(\m f(\m x,\m u)-\m f(\hat{\m x},\m u)),\m x-\hat{\m x}\rangle\leq \gamma_s \norm{\m x - \hat{\m x}}_2^{2}
		\end{align}
		for all $(\m x, \m u), (\hat{\m x}, \m u)\in \mathbf{\Omega}$ with
		\begin{align}\label{eq:osl_bound_for_gamma_s}
			\hspace{-0.3cm}	\gamma_s = \Bigg(\max_{(\m x, \m u)\in \m \Omega} \sum_{i, j\in \mbb{I}(g)} \Bigg|\sum_{k\in \mbb{I}(g)} G_{(i,k)} \dfrac{\partial f_k}{\partial x_j}(\m x, \m u)\Bigg|^2\hspace{0.05cm}\Bigg)^{\hspace{-0.05cm}\!1/2}.
		\end{align}
	\end{subequations}
\end{theorem}
\begin{proof} { 
		We simply find an upper bound for $\bar \gamma$ given in \eqref{eq:one_sided_Lipschitz_upper}. For fixed, but arbitrary $(\m x, \m u)\in \m \Omega$, let $\m v\neq 0$ be an eigenvector for the maximal eigenvalue $\lambda_{\max}(\m x, \m u)$ of $\frac{1}{2}\left(\m \Xi + {\m \Xi}^{\top}\right)$. From the eigenvalue equation and the fact that $\langle \m \Xi \m v, \m v \rangle = \langle \m \Xi^{\top} \m v, \m v \rangle$ we can deduce that $\lambda_{\max}(\m x, \m u) \norm{\m v}_2^2 = \langle \m \Xi(\m x, \m u)  \m v, \m v \rangle$
		and hence
		\begin{align}\label{eq:proof_osl_eq1}
			\lambda_{\max}(\m x, \m u) \norm{\m v}_2^2 \leq \norm{\m\Xi(\m x, \m u)}_2 \norm{\m v}^2_2.
		\end{align}
		To proceed, note that \eqref{eq:osl_bound_for_gamma_s} holds if $\bar \gamma \leq 0$. Otherwise, there is some point $(\m x^*, \m u^*)$ at which $\bar \gamma = \lambda_{\max}(\m x^*, \m u^*) > 0$. We can square \eqref{eq:proof_osl_eq1} and recall the definition of $\m\Xi(\cdot)$ to get
		\begin{align*}
			\lambda_{\max}(\m x^*, \m u^*)^2 \leq \max_{(\m x, \m u)\in \m \Omega} \sum_{i, j\in \mbb{I}(g)} \Bigg|\sum_{k\in \mbb{I}(g)} G_{(i,k)} \dfrac{\partial f_k}{\partial x_j}(\m x, \m u)\Bigg|^2, 
		\end{align*}
		and thus the desired result.}
\end{proof}
\vspace{-0.3cm}
Another less straightforward, that is potentially less conservative approach to calculate $\gamma_s$ than \eqref{eq:osl_bound_for_gamma_s} is to make use one of the consequence from \textit{Gershgorin's circle theorem}. In essence, this particular theorem guarantees that each eigenvalue of a matrix is always confined by a disk characterized by the diagonal and non-diagonal entries of that matrix \cite{barany2017gershgorin}. The next proposition recapitulates this approach to compute an upper bound for the greatest eigenvalue of a symmetric matrix.
\vspace{-0.16cm}
\begin{myprs}\label{prs:gershgorin_max_eigenvalue}
	For any $\m \Psi \in \mbb{S}^n$, the following inequality holds
	\begin{align}
		\lambda_{\mathrm{max}}\left(\m \Psi \right) \leq \max_{i\in \mbb{I}(n)}\Bigg(\Psi_{(i,i)} + \sum_{j \in\mbb{I}(n)\setminus i}\abs{\Psi_{(i,j)}}\Bigg).\label{prs:gershgorin_inequality_max_eigenvalue}
	\end{align}
\end{myprs}
\begin{proof}
	\begin{subequations}
		Suppose $\lambda\in\mbb{R}$ be an eigenvalue of $\m \Psi \in \mbb{S}^n$ and $\m v\in\mbb{R}^n$ be the corresponding eigenvector. Since $\m v\neq \m 0$, then there exists $v_i$ such that $\abs{v_j} \leq \abs{v_i}$ for all $j\in \mbb{I}(n)\setminus i$. We then perform the following transformation: $\m v \leftarrow \mathrm{sgn}(v_i)\frac{\m v}{\abs{v_i}}$. Here, $\mathrm{sgn}(\cdot)$ denotes the sign function.
		Since the new $\m v$ also satisfies $(\m \Psi-\lambda\mI)\m v = \m 0$, then at the $i$-th row we have
		\begin{align}
			\lambda v_i = \sum_{j\in \mbb{I}(n)}\Psi_{(i,i)} v_j = \Psi_{(i,i)} v_i + \sum_{j\in \mbb{I}(n)\setminus i}\Psi_{(i,j)} v_j. \label{eq:gershgorin_inequality_max_proof_1}
		\end{align}
		Observe that the prior transformation leads to $v_i = 1$ and $\abs{v_j} \leq 1$ for all $j\in \mbb{I}(n)\setminus i$. Therefore, from \eqref{eq:gershgorin_inequality_max_proof_1} and applying the triangle inequality, we obtain 
		\begin{align}
			\lambda - \Psi_{(i,i)} &\leq \sum_{j\in \mbb{I}(n)\setminus i}\abs{\Psi_{(i,j)}}\abs{ v_j} \leq \sum_{j\in \mbb{I}(n)\setminus i}\abs{\Psi_{(i,j)}}. \label{eq:gershgorin_inequality_max_proof_2}
		\end{align}
		As \eqref{eq:gershgorin_inequality_max_proof_2} applies to any eigenvalue of $\m \Psi$, then \eqref{prs:gershgorin_inequality_max_eigenvalue} holds.
	\end{subequations}
\end{proof}
\vspace{-0.16cm}
Proposition \ref{prs:gershgorin_max_eigenvalue} provides an amenable way which can be used for computing OSL constant $\gamma_s$ provided that $\m \Psi = \frac{1}{2}\left(\mathbf{\Xi} (\m x,\m u) + \mathbf{\Xi}^{\top} (\m x,\m u)\right)$. That is,
\begin{align*}
	\gamma_s =  \max_{i\in \mbb{I}(n)}\Bigg(\max_{(\m{x},\m u)\in \mathbf{\Omega}}\Bigg(\Psi_{(i,i)} + \sum_{j \in\mbb{I}(n)\setminus i}\abs{\Psi_{(i,j)}}\Bigg)\hspace{-0.1cm}\Bigg).
\end{align*}
Note that in \eqref{prs:gershgorin_inequality_max_eigenvalue} it is possible for $\lambda_{\mathrm{max}}\left(\m \Psi \right)$ to be nonpositive assuming that $\m \Psi $ is diagonally dominant with $\Psi_{(i,i)}\leq 0$ for each $i\in\mbb{I}(n)$.  
In the following theorem, we develop an alternative method, aside from \eqref{prs:gershgorin_inequality_max_eigenvalue}, to provide an upper bound for the greatest eigenvalue of any symmetric matrix.
\vspace{-0.16cm}
\begin{theorem}\label{thm:vu_max_eigenvalue}
	It holds for any $\m \Psi \in \mbb{S}^n$ that
	\begin{subequations}
		\begin{align}
			\lambda_{\mathrm{max}}\left(\m \Psi \right) \leq \max_{i\in \mbb{I}(n)}\left(\Psi_{(i,i)} +\zeta_n\max_{j\in \mbb{I}(n)\setminus i}\abs{\Psi_{(i,j)}}\right),\label{thm:vu_inequality_max_eigenvalue}
		\end{align}	
		where $\zeta_n \in \mbb{R}_{++}$ is a scalar that depends on the dimension $n$ and the optimal value of the following maximization problem { 
			\begin{align}
				\zeta_n = \hspace*{-0.0cm}\max_{\m v}\;\; & \tfrac{1}{v_i}-1 \label{eq:vu_inequality_max_eigenvalue_const_1}\\
				\subjectto  \;\;	 & \sum_{j \in\mbb{I}(n)} \abs{v_j} = 1, \;\m v\in\mbb{R}^n \label{eq:vu_inequality_max_eigenvalue_const_2}\\
				&v_i > 0,\;\abs{ v_j}  \leq v_i,\;\forall j\in \mbb{I}(n)\setminus i.\label{eq:vu_inequality_max_eigenvalue_const_3}
		\end{align}}
	\end{subequations}
\end{theorem}
\vspace{-0.3cm}
\begin{proof}
	\begin{subequations}
		  {Suppose that $\lambda\in\mbb{R}$ is an eigenvalue of $\m \Psi \in \mbb{S}^n$ and $\m v\in\mbb{R}^n$ is the corresponding eigenvector.} Since $\m v\neq \m 0$, there exists $v_i$ such that $\abs{v_j} \leq \abs{v_i}$ for every $j\in \mbb{I}(n)\setminus i$.   {Next, the following transformation is applied to $\m v$}  
		\begin{align}
			\m v \leftarrow \mathrm{sgn}(v_i)\frac{\m v}{\sum_{j \in\mbb{I}(n)} \abs{v_j}},\label{eq:vu_thm_proof_1}
		\end{align}
		  {where $\sum_{j \in\mbb{I}(n)} \abs{v_j}= 1$.} As $(\m \Psi-\lambda\mI)\m v = \m 0$ for such $\m v$ described in \eqref{eq:vu_thm_proof_1}, then at the $i$-th row we have
		\begin{align}
			\lambda v_i = \sum_{j\in \mbb{I}(n)}\Psi_{(i,i)} v_j = \Psi_{(i,i)} v_i + \sum_{j\in \mbb{I}(n)\setminus i}\Psi_{(i,j)} v_j. \label{eq:vu_thm_proof_2} 
		\end{align}
		Notice that the transformation \eqref{eq:vu_thm_proof_1} yields $v_i > 0$ and $\abs{v_j} \leq {v_i}$ for all $j\in \mbb{I}(n)\setminus i$. Thus, from \eqref{eq:vu_thm_proof_2} and applying the triangle inequality, we obtain
		\begingroup
		\allowdisplaybreaks
		\begin{align}
			\lambda - \Psi_{(i,i)} &\leq \frac{1}{v_i}\Bigg(\sum_{j\in \mbb{I}(n)\setminus i}\abs{\Psi_{(i,j)}}\abs{ v_j}\Bigg)\nonumber \\
			&\leq \frac{1}{v_i}\Bigg(\max_{j\in \mbb{I}(n)\setminus i}\abs{\Psi_{(i,j)}}\sum_{j\in \mbb{I}(n)\setminus i}\abs{ v_j}\Bigg)\nonumber \\
			&=  \frac{\sum_{j\in \mbb{I}(n)}\abs{ v_j}-v_i}{v_i}\max_{j\in \mbb{I}(n)\setminus i}\abs{\Psi_{(i,j)}}\nonumber \\
			&=  \Big(\frac{1}{v_i}-1\Big)\max_{j\in \mbb{I}(n)\setminus i}\abs{\Psi_{(i,j)}}.\label{eq:vu_thm_proof_3} 
		\end{align}
		\endgroup
		Finally, since \eqref{eq:vu_thm_proof_3} is valid for every eigenvalue of $\m\Psi$, then we eventually obtain \eqref{thm:vu_inequality_max_eigenvalue} where $\zeta_n = \frac{1}{v_i}-1$ is computed through solving optimization problem specified in \eqref{eq:vu_inequality_max_eigenvalue_const_1}-\eqref{eq:vu_inequality_max_eigenvalue_const_3}. 
	\end{subequations}
\end{proof}
\vspace{-0.15cm}{ 
	Notice that finding $\zeta_n$ is a difficult task on its own since the problem given in \eqref{eq:vu_inequality_max_eigenvalue_const_1}-\eqref{eq:vu_inequality_max_eigenvalue_const_3} is not convex. 
	To that end, in the following proposition we provide an approach to compute $\zeta_n$ via a convex optimization problem, which later used to analytically determine $\zeta_n$.
	\vspace{-0.15cm}
	\begin{myprs}\label{prs:computing_zeta}
		The optimal value of the problem described in \eqref{eq:vu_inequality_max_eigenvalue_const_1}-\eqref{eq:vu_inequality_max_eigenvalue_const_3} is $\zeta_n^* = n-1$ where $n$ is the dimension of $\m \Psi$.
	\end{myprs}
\begin{proof}
	First, observe that since $v_i > 0$, then maximizing $\tfrac{1}{v_i}$ pushes $v_i$ to be as close as possible to zero. Hence, it is equivalent to minimizing $v_i$ given that $v_i > 0$. As such, $\zeta_n^*$ can be computed as $\zeta_n^* = \tfrac{1}{v_i^*}-1$.
	This process reformulates \eqref{eq:vu_inequality_max_eigenvalue_const_1}-\eqref{eq:vu_inequality_max_eigenvalue_const_3} into a minimization problem
	\begin{subequations}\label{eq:computing_zeta}
		\begin{align}
			v_i^* = \hspace*{-0.0cm}\min_{\m v\in\mbb{R}^n}\;\; & v_i \label{eq:computing_zeta_1}\\
			\subjectto  \;\; &	 \eqref{eq:vu_inequality_max_eigenvalue_const_2},\;\eqref{eq:vu_inequality_max_eigenvalue_const_3}.\label{eq:computing_zeta_2}
		\end{align}
	\end{subequations}
Suppose that there exists an index $k$ such that $v_k \leq 0$. Since this implies $v_k < v_i$, then $v_k$ would be the minimizer. However, this is not possible since being a minimizer enforces the constraint $v_k > 0$. This process can be repeated for other indices in $\mbb{I}(n)$. From this observation, once can relax the absolute values appearing in \eqref{eq:vu_inequality_max_eigenvalue_const_2} and \eqref{eq:vu_inequality_max_eigenvalue_const_3} without sacrificing optimality. The resulting problem becomes 
\begin{subequations}\label{eq:computing_zeta2}
\begin{align}
	\min_{\m v\in\mbb{R}^n}\;\; & v_i \label{eq:computing_zeta2_1}\\
\subjectto  \;\;	 & \sum_{j \in\mbb{I}(n)} {v_j} = 1, \; \label{eq:computing_zeta2_2}\\
&v_i > 0,\;{ v_j}  \leq v_i,\;\forall j\in \mbb{I}(n)\setminus i.\label{eq:computing_zeta2_3}
\end{align}
\end{subequations}
Without loss of generality, let us assume for now that $i = 1$. For a sufficiently small $\epsilon >0$, the Lagrangian of \eqref{eq:computing_zeta2} is
\begin{align*}
\mathcal{L}(\m v,\m \lambda,\mu) &= v_1 + \lambda_1\left(\epsilon - v_1\right) +\sum_{j\in \mbb{I}(n)\setminus 1} \lambda_j\left(v_j - v_1\right) \\&\quad + \mu \left(\sum_{j \in\mbb{I}(n)} v_j- 1\right),
\end{align*}
where $\m \lambda\in\mbb{R}^{n}$ and $ \mu\in\mbb{R}$ are the Lagrange multipliers for the inequality and equality constraints. Since the Lagrangian is differentiable, the gradient of it with respect to $\m v$ is
\begin{align*}
\nabla_{v} \mathcal{L}(\m v,\m \lambda,\mu) = \bmat{1 - \sum_{j\in \mbb{I}(n)} \lambda_j + \mu \\ \lambda_2 + \mu \\ \vdots \\ \lambda_n + \mu }.
\end{align*}
Observe that, since the primal problem (problem \eqref{eq:computing_zeta2}) is convex and the equality constraint is affine, then any points $\left(\m v^*,\m \lambda^*,\mu^*\right)$ satisfying the Karush-Kuhn-Tucker (KKT) conditions are primal and dual optimal with zero duality gap \cite{Boyd2004Convex}. To that end, it is sufficient to find $\left(\m v^*,\m \lambda^*,\mu^*\right)$ that satisfy the KKT conditions. First, let $v^*_i = v^*_j$ for all $i\neq j$. From \eqref{eq:computing_zeta2_2}, we get $v^*_i = \frac{1}{n}$ for all $i\in\mbb{I}(n)$. This particular choice of $\m v^*$ satisfies all the constraints in \eqref{eq:computing_zeta2} (in this case, we can always set $\epsilon = \frac{1}{2n}$). Now, from the complementary slackness, it can be inferred that $\lambda^*_1 = 0$  and $\lambda^*_j > 0$ for all $j\in \mbb{I}(n)\setminus 1$. By setting $\nabla_{v} \mathcal{L}(\m v^*,\m \lambda^*,\mu^*) = 0$, the following set of $n$ equations with $n$ variables are obtained
\begin{align*}
	 1 - \sum_{j\in \mbb{I}(n)\setminus 1} \lambda^*_j + \mu^* &= 0,\; 
	 \lambda^*_2 + \mu^* = 0,\; \hdots, \;
	 \lambda^*_n + \mu^* = 0.
\end{align*}
One immediate solution is to set $\mu^* = -\frac{1}{n}$ and $\lambda^*_j = \frac{1}{n}$ for all $j\in \mbb{I}(n)\setminus 1$. This shows that $ v^*_1 = \frac{1}{n}$. By repeating these steps for $i = 2,3,\hdots,n$, one gets $ v^*_i = \frac{1}{n}$ and finally $\zeta_n^* = n-1$.
\end{proof}
\vspace{-0.15cm}
	In summary, this section proposes three different approaches to compute OSL constant, all of which are based on computing an upper bound for the maximum eigenvalue of matrix $\mathbf{\Xi}  (\cdot)$. Define $\gamma_{s_1}$, $\gamma_{s_2}$, and $\gamma_{s_3}$ as OSL constants obtained from \eqref{eq:osl_bound_for_gamma_s}, \eqref{prs:gershgorin_inequality_max_eigenvalue}, and \eqref{thm:vu_inequality_max_eigenvalue} respectively. Since Theorem \ref{cor:osl_bounds} gives a nonnegative upper bound towards maximum eigenvalue, then we have  $\gamma_{s_2}\leq \gamma_{s_1}$ and $\gamma_{s_3}\leq \gamma_{s_1}$.
	However, further investigation is required to determine the relation between $\gamma_{s_2}$ and $\gamma_{s_3}$.  }

\vspace{-0.3cm}
\subsection{Quadratic Inner-Boundedness (QIB)}\label{ssec:qib}
The concept of QIB has been extensively used alongside OSL condition for a wide variety of observer design with applications ranging from state estimation to feedback stabilization \cite{Abbaszadeh2010,zhang2012full,song2015robust,liu2014static,Rastegari2019,Gholami2019 }. Analogously to OSL, the constants $\gamma_{q1}$ and $\gamma_{q2}$ for QIB given in Tab. \ref{tab:nonlinear_class} may also be any real numbers. The following theorem summarizes our result for computing QIB constants.
\vspace{-0.1cm}
\begin{theorem}\label{prs:quadratic_inner_bounded}
	The nonlinear function $\m f :\mathbb{R}^n\times \mathbb{R}^m\rightarrow \mathbb{R}^g$ in \eqref{eq:gen_dynamic_systems} is locally QIB in $\mathbf{\Omega}$ such that for any $(\m x, \m u), (\hat{\m x}, \m u)\in \mathbf{\Omega}$ the following holds
	\vspace{-0.1cm}
	\begingroup
	\allowdisplaybreaks
	\begin{subequations}\label{eq:quadratic_inner_bounded_theorem}
		\begin{align}
			&\langle\m G(\m f(\m x,\m u)-\m f(\hat{\m x},\m u)),\m G(\m f(\m x,\m u)-\m f(\hat{\m x},\m u))\rangle\leq \nonumber \\
			& \;\gamma_{q1} \norm{\m x - \hat{\m x}}_2^{2}+\gamma_{q2}\langle\m G(\m f(\m x,\m u)-\m f(\hat{\m x},\m u)),\m x-\hat{\m x}\rangle, \label{eq:quadratic_inner_bounded}
		\end{align}
		where for $\epsilon_1,\epsilon_2\in \mbb{R}_{+}$, $\gamma_{q2} = \epsilon_2-\epsilon_1$ and $\gamma_{q1}$ is specified as
		\begin{align}
			\gamma_{q1} = \hspace{-0.0cm}\epsilon_1\bar{\gamma}-\epsilon_2\barbelow{\gamma}+\max_{(\m{x},\m u)\in  \mathbf{\Omega}}\hspace{-0.0cm}\sum_{i\in\mbb{I}(n)}\hspace{-0.0cm}\norm{\nabla_{\hspace{-0.05cm}x} \hspace{0.04cm}\xi_i(\m x,\m u)}_2^2,\hspace{-0.00cm} \label{eq:gamma_q1_qib}
		\end{align}
		where $\xi_i(\m x,\m u) := \sum_{j\in\mbb{I}(g)} G_{(i,j)}f_j(\m x, \m u)$ for $i\in\mbb{I}(n)$ and
		$\bar{\gamma}$ and $\barbelow{\gamma}$ are the optimal values of \eqref{eq:one_sided_Lipschitz_upper} and \eqref{eq:one_sided_Lipschitz_lower}.
	\end{subequations}
	\endgroup
\end{theorem}
\begin{proof}
	\begin{subequations}
		Applying again Lemma \ref{lem:key} with $\m w = \m G(\m f(\m x, \m u)-\m f(\hat{\m x}, \m u))$, we can deduce the following result
		\begin{align}
			&\langle\m G(\m f(\m x,\m u)-\m f(\hat{\m x},\m u)),\m G(\m f(\m x,\m u)-\m f(\hat{\m x},\m u))\rangle\leq\nonumber \\ &\quad\quad 
			\Bigg(\max_{(\m{x},\m u)\in  \mathbf{\Omega}}\hspace{-0.0cm}\sum_{i\in\mbb{I}(n)}\hspace{-0.0cm}\norm{\nabla_{\hspace{-0.05cm}x} \hspace{0.04cm}\xi_i(\m x,\m u)}_2^2\hspace{-0.00cm} \Bigg)\hspace{-0.05cm}\norm{\m x - \hat{\m x}}_2^{2}.  \label{eq:qib_proof_5}
		\end{align}
		From \eqref{eq:one_sided_Lipschitz}, we get for $\epsilon_1,\epsilon_2\in \mbb{R}_{+}$
		\begin{align}
			0 &\leq \epsilon_1\bar{\gamma} \norm{\m x - \hat{\m x}}_2^{2} -\epsilon_1\langle\m G(\m f(\m x,\m u)-\m f(\hat{\m x},\m u)),\m x-\hat{\m x}\rangle\label{eq:qib_proof_6}\\
			0 &\leq -\epsilon_2\barbelow{\gamma} \norm{\m x - \hat{\m x}}_2^{2} +\epsilon_2\langle\m G(\m f(\m x,\m u)-\m f(\hat{\m x},\m u)),\m x-\hat{\m x}\rangle.\label{eq:qib_proof_7}
		\end{align} 
		Combine these two inequalities by multiplying both with $(-1)$ and adding them to
		\begin{align}
			\begin{split}
				- &(\epsilon_2 - \epsilon_1) \langle \m G(\m f(\m x,\m u)
				-\m f(\hat{\m x},\m u)),\m x-\hat{\m x}\rangle \\
				&- (\epsilon_1 \bar{\gamma} - \epsilon_2 \barbelow{\gamma}) \norm{\m x - \hat{\m x}}_2^2 \leq 0
			\end{split}\label{eq_qib_proof_8}
		\end{align}
		Finally, we can deduce \eqref{eq:quadratic_inner_bounded} by using  \eqref{eq:qib_proof_5}, \eqref{eq_qib_proof_8} and noting that $k\norm{\m x - \hat{\m x}}_2^2$ equals
		\begin{align*}
			&k\norm{\m x - \hat{\m x}}_2^2 + (\eps_2\barbelow{}{\gamma} - \eps_1\bar{\gamma}) \langle\m G(\m f(\m x,\m u)-\m f(\hat{\m x},\m u)),\m x-\hat{\m x}\rangle \\
			&+ (\eps_1 \bar{\gamma}-\eps_2\barbelow{\gamma}) \langle\m G(\m f(\m x,\m u)-\m f(\hat{\m x},\m u)),\m x-\hat{\m x}\rangle\\
			&\leq \gamma_{q_1} \norm{\m x - \hat{\m x}}_2^2 + (\eps_2\barbelow{\gamma} - \eps_1\bar{\gamma}) \langle\m G(\m f(\m x,\m u)-\m f(\hat{\m x},\m u)),\m x-\hat{\m x}\rangle
		\end{align*}
		with $\gamma_{q_1}$ given in
		\eqref{eq:gamma_q1_qib}.
	\end{subequations}
\end{proof}
\vspace{-0.22cm}
This result allows the QIB constants $\gamma_{q1}$ and $\gamma_{q2}$ to be parametrized with nonnegative variables $\epsilon_1$ and $\epsilon_2$, hence giving one kind of degree of freedom that can be useful for observer/controller design. 
From \eqref{eq:qib_proof_5} and using a similar approach as in Corollary \ref{cor:Lipschitz}, one can also verify that
\begin{align*}
	\max_{(\m{x},\m u)\in  \mathbf{\Omega}}\hspace{-0.0cm}\sum_{i\in\mbb{I}(n)}\hspace{-0.0cm}\norm{\nabla_{\hspace{-0.05cm}x} \hspace{0.04cm}\xi_i(\m x,\m u)}_2^2 \leq \hspace{-0.0cm}\sum_{i\in\mbb{I}(n)}\hspace{-0.0cm}\max_{(\m{x},\m u)\in  \mathbf{\Omega}}\norm{\nabla_{\hspace{-0.05cm}x} \hspace{0.04cm}\xi_i(\m x,\m u)}_2^2,
\end{align*}
which prospectively allows $\gamma_{q1}$ to be computed more efficiently in a distributive manner.

\vspace{-0.4cm}
\subsection{Quadratic Boundedness (QB)}\label{ssec:qb}
Unlike the other function sets, QB condition is mainly utilized to generate conditions for designing stabilizing feedback control actions, with a much less common applications in observer design \cite{Guo2017}. This condition is introduced in \cite{Siljak2000} and has been widely applied since then for numerous stabilization purposes \cite{Siljak2002}.
  {Following \cite{Siljak2000,Siljak2002}, the subsequent assumptions for $\m f(\cdot)$ are considered: $\m f(\cdot)$ is independent of $\m u$ and $\m f(0) = 0$.} As such the following result is established.
\vspace{-0.15cm}{ 
	\begin{theorem}\label{prs:quadratic_bounded}
		The nonlinear function $\m f :\mathbb{R}^n\rightarrow \mathbb{R}^g$ in \eqref{eq:gen_dynamic_systems} is locally QB in $\mathbfcal{X}$ such that 
		\begin{subequations}
			\begin{align}
				&\langle \m f(\m x),\m f(\m x)\rangle \leq  \m x^{\top}\m \Gamma^{\top}\m \Gamma \m x, \label{eq:quadratic_bounded}
				\vspace{-0.2cm}
			\end{align}
			for $\m x\in \mathbfcal{X}$ where $\m \Gamma\in\mbb{R}^{n\times n}$ can be computed as follow
			
			\begin{align}
				\vspace{-0.05cm}
				\hspace{-0.0cm}\m \Gamma = \mathrm{Diag}\hspace{-0.00cm}\left(\hspace{-0.05cm}\left\{\max_{\m{x}\in \mathbfcal{X}}  \sqrt{n\hspace{-0.02cm}\sum_{i\in \mbb{I}(g)}\hspace{-0.1cm}\left(\dfrac{\partial f_i}{\partial x_{j}}(\m x)\right)^2} \hspace{0.05cm}\right\}^{\hspace{-0.05cm}n}_{\hspace{-0.05cm}j = 1}\right).\label{eq:quadratic_bounded_2}
			\end{align}
		\end{subequations}
\end{theorem}}
\vspace{-0.43cm}
\begin{proof}
	Let $\m x,\hat{ \m  x}\in \mathbfcal{X}$ and define two functions $\phi:\mathbb{R}\rightarrow \mathbb{R}$ and $\m g:\mathbb{R}\rightarrow \mathbb{R}^n$ as follows
	\begin{subequations}
		\begin{align*}
			\phi(\tau) &:= \langle {\m f}(\m x)-{\m f}(\hat{ \m  x}),\m f(\m g(\tau))\rangle,\;\;
			\m g(\tau) := \hat{\m x} + \tau (\m x -\hat{\m x} ).
		\end{align*}
		Since $\phi(\cdot)$ is continuous on $[0,1]$ and differentiable on $(0,1)$, then from mean value theorem there exists such $\bar{\tau}\in(0,1)$ that
		\begin{align}
			\frac{d\phi}{d\tau}(\bar{\tau}) &= \frac{\phi(1)-\phi(0)}{1-0} = \phi(1)-\phi(0). \label{eq:qb_proof_1} 
		\end{align}
		For brevity, define $\tilde{\m f}(\m x,\hat{ \m  x}) :={\m f}(\m x)-{\m f}(\hat{ \m  x})$. The left-hand side of \eqref{eq:qb_proof_1} is equivalent to
		\begin{align}
			\frac{d\phi}{d\tau}(\bar{\tau}) &= \hspace{-0.05cm}\sum_{i\in\mbb{I}(g)} \hspace{-0.05cm}\tilde{ f}_i(\m x,\hat{ \m  x})\frac{d f_i\left(\m g(\bar{\tau})\right)}{d\tau}\nonumber \\
			&= \hspace{-0.05cm}\sum_{i\in\mbb{I}(g)} \hspace{-0.05cm}\tilde{ f}_i(\m x,\hat{ \m  x})\hspace{-0.05cm}\left(\sum_{j\in \mbb{I}(n)}\hspace{-0.05cm}\dfrac{\partial f_i\left(\m g(\bar{\tau})\right)}{\partial g_{j}}\cdot\dfrac{d g_{j}}{d \tau}\hspace{-0.05cm}\right).  \label{eq:qb_proof_2}
		\end{align}
		Since $\frac{d g_{j}}{d \tau} = x_j-\hat{x}_j$ and $\m g(\bar{\tau})\in \left(\hat{\m x},\m x \right)$ for which $\bar{\tau}\in(0,1)$ implies that there exists $\m z\in \left(\hat{\m x},\m x \right)$, \eqref{eq:qb_proof_2} is equivalent to
		\begin{align}
			\frac{d\phi}{d\tau}(\bar{\tau})&= \hspace{-0.05cm}\sum_{i\in\mbb{I}(g)} \hspace{-0.05cm}\tilde{ f}_i(\m x,\hat{ \m  x})\hspace{-0.05cm}\left(\sum_{j\in \mbb{I}(n)}\hspace{-0.05cm}\dfrac{\partial f_i\left(\m z\right)}{\partial x_{j}}\left(x_j-\hat{x}_j\right)\hspace{-0.05cm}\right).  \label{eq:qb_proof_3}
		\end{align}
		  {By realizing that the right hand side of \eqref{eq:qb_proof_1} is}
		\begin{align*}
			\phi(1)-\phi(0) \hspace{-0.05cm}&= \hspace{-0.05cm}\langle \tilde{ \m f}(\m x,\hat{ \m  x}),\m f(\m g(1))-\m f(\m g(0))\rangle
			\hspace{-0.05cm}=\hspace{-0.05cm} \norm{\tilde{ \m f}(\m x,\hat{ \m  x})}_2^2,
		\end{align*}
		then from the above, \eqref{eq:qb_proof_1}, and \eqref{eq:qb_proof_3}, one can verify that
		\begin{align}
			&\norm{\tilde{ \m f}(\m x,\hat{ \m  x})}_2^2 = \hspace{-0.1cm}\sum_{i\in\mbb{I}(g)} \hspace{-0.05cm}\tilde{ f}_i(\m x,\hat{ \m  x})\hspace{-0.1cm}\left(\sum_{j\in \mbb{I}(n)}\hspace{-0.05cm}\dfrac{\partial f_i\left(\m z\right)}{\partial x_{j}}\left(x_j-\hat{x}_j\right)\hspace{-0.1cm}\right) \nonumber \\ 
			&\quad\leq \hspace{-0.1cm}\sum_{i\in\mbb{I}(g)} \hspace{-0.00cm}\abs{\tilde{ f}_i(\m x,\hat{ \m  x})\hspace{-0.1cm}\left(\sum_{j\in \mbb{I}(n)}\hspace{-0.05cm}\dfrac{\partial f_i\left(\m z\right)}{\partial x_{j}}\left(x_j-\hat{x}_j\right)\hspace{-0.1cm}\right)\hspace{-0.05cm}} \nonumber \\ 
			&\quad\leq \norm{\tilde{ \m f}(\m x,\hat{ \m  x})}_2 \left(\sum_{i\in\mbb{I}(g)}\abs{\sum_{j\in \mbb{I}(n)}\hspace{-0.05cm}\dfrac{\partial f_i\left(\m z\right)}{\partial x_{j}}\left(x_j-\hat{x}_j\right)\hspace{-0.00cm}}^2\right)^{\hspace{-0.05cm}\!1/2}, \nonumber
		\end{align}
		  {which implies that}
		\begingroup
		\allowdisplaybreaks
		\begin{align}
			&\tilde{ \m f}(\m x,\hat{ \m  x})^{\top}\tilde{ \m f}(\m x,\hat{ \m  x}) \leq \sum_{i\in\mbb{I}(g)}\hspace{0.05cm}\abs{\sum_{j\in \mbb{I}(n)}\hspace{-0.05cm}\dfrac{\partial f_i\left(\m z\right)}{\partial x_{j}}\left(x_j-\hat{x}_j\right)\hspace{-0.00cm}}^2\nonumber \\ 
			&\quad\leq \sum_{i\in\mbb{I}(g)}\hspace{-0.1cm}\left(\sum_{j\in \mbb{I}(n)}\hspace{-0.00cm}\abs{\dfrac{\partial f_i\left(\m z\right)}{\partial x_{j}}\left(x_j-\hat{x}_j\right)}\hspace{-0.02cm}\right)^{\hspace{-0.1cm}2}\nonumber \\ 
			&\quad\leq \sum_{i\in\mbb{I}(g)}\hspace{-0.05cm}n\hspace{-0.05cm}\left(\sum_{j\in \mbb{I}(n)}\hspace{-0.1cm}\left(\dfrac{\partial f_i\left(\m z\right)}{\partial x_{j}}\right)^2\hspace{-0.02cm}\left(x_j-\hat{x}_j\right)^2\hspace{-0.05cm}\right)\nonumber \\ 
			&\quad\leq \max_{\m{x}\in \mathbfcal{X}} \sum_{i\in\mbb{I}(g)}\hspace{-0.00cm}\hspace{-0.02cm}\sum_{j\in \mbb{I}(n)}\hspace{-0.05cm}n\hspace{-0.05cm}\left(\dfrac{\partial f_i\left(\m x\right)}{\partial x_{j}}\right)^2\hspace{-0.02cm}\left(x_j-\hat{x}_j\right)^2\hspace{-0.05cm}\nonumber \\ 
			&\quad\leq \sum_{j\in \mbb{I}(n)}\hspace{-0.1cm}\left(\hspace{-0.05cm}\max_{\m{x}\in \mathbfcal{X}}\sum_{i\in\mbb{I}(g)}\hspace{-0.05cm}n\hspace{-0.05cm}\left(\dfrac{\partial f_i\left(\m x\right)}{\partial x_{j}}\right)^{\hspace{-0.1cm}2}\hspace{-0.00cm}\hspace{-0.00cm}\right)\hspace{-0.05cm}\left(x_j-\hat{x}_j\right)^2\nonumber \\
			&\quad=(\m x-\hat{  \m x})^{\top}\m \Gamma^{\top}\m \Gamma (\m x-\hat{  \m x}), \label{eq:qb_proof_4}
		\end{align}
		\endgroup
		where the diagonal matrix $\m \Gamma$ is described in \eqref{eq:quadratic_bounded_2}. Since $\m f(0) = 0$, then from \eqref{eq:qb_proof_4} one can immediately obtain \eqref{eq:quadratic_bounded}.
	\end{subequations}
\end{proof}
\vspace{-0.15cm}
As seen in \eqref{eq:quadratic_bounded_2}, the matrix $\m \Gamma$ is diagonal. As such, one can indeed solve the global maximization problem for each diagonal entry of $\m \Gamma$ individually in a distributive manner without trading computational time with conservativeness. 
\vspace{-0.1cm}
\begin{myrem}(Difference between QB and QIB)
	Although the last two function sets---QB and QIB---seemingly share canny resemblance due to their terminology, the two sets are drastically different. Admittedly, the terminology is somewhat confusing since it suggests a connection between the two. QB is simply a bound on the values the function can attain, implying that the norm of $\m f(\cdot)$ is bounded by a quadratic function. On the other hand, QIB is akin to the Lipschitz condition, since it compares the values of $\m f(\cdot)$ at two different points. 
\end{myrem}
\vspace{-0.1cm}

{ 
	\vspace{-0.35cm}
	\subsection{Remark on Improving Scalability and New Insights}\label{ssec:lip_qib_relation}
	So far we have presented systematic methods of transforming NDS parameterization as global maximization problems for various function sets. 
	  {Herein, we \textit{(a)} discuss some methods to improve scalability of the proposed approaches for NDS parameterization and \textit{(b)} review the relation between Lipschitz continuous and QIB function sets and provide some necessary corrections.}
	
	We now discuss important remarks and strategies for improving the scalability when solving global optimization problems for NDS parameterization---listed as follows. 
	\begin{itemize}[leftmargin=*]
		\item The first strategy, as mentioned in Remark \ref{rem:search_space_reduction}, can be referred to as \textit{search space reduction}. For instance, in the context of Lipschitz parameterization where not every components of $\m{x}$ and $\m u$ appears as an argument in $\nabla_{\hspace{-0.05cm}x} f_i(\m{x},\m u)$, one can neglect these components from $\m \Omega$ such that it is sufficient to solve
		$\max_{(\m{x},\m u)\in {\mathbf{\Omega}_{\m r_i}}}\norm{\nabla_{\hspace{-0.05cm}x} f_i(\m{x},\m u)}_2^2$ where $\mathbf{\Omega}_{\m r_i}$ denotes the reduced search space.  
		\item The second strategy exploits the mathematical structure of $\m f(\cdot)$.   {Most of NDS representing many modern infrastructures are comprised of several individual subsystems that are intertwined together, forming complex networked NDS.} The majority of these subsystems, to some extent, share similar forms of nonlinear dynamics. As an illustration, consider a network of highway traffic \cite{nugroho2018,nugroho2018journal}. Assuming a free-flow condition, the nonlinearity on the $i$-th highway segment not connected to ramps is given as $f_i(\m x) = \delta\left(x^2_i-x^2_{i-1}\right)$ where $\delta$ is a constant. If the $i+j$-th highway segment is also not connected to ramps for some index $j$, then $f_{i+j}(\m x) = \delta\left(x^2_{i+j}-x^2_{i+j-1}\right)$ and thus from \eqref{eq:gamma_Lipschitz_2}, it suffices to solve $\max_{(\m{x},\m u)\in \mathbf{\Omega}_i}\norm{\nabla_{\hspace{-0.05cm}x} f_i(\m{x},\m u)}_2^2$ once since the optimal values for $i$ and $i+j$ are the same (assuming that $\mathbf{\Omega}_i = \mathbf{\Omega}_{i+j}$ for some index $j$).
		\item The third strategy employs a distributed computing paradigm, where in the context of computing Lipschitz constant, its computational burden is shared within several computers (or machines) as the problem $\max_{(\m{x},\m u)\in \mathbf{\Omega}}\norm{\nabla_{\hspace{-0.05cm}x} f_i(\m{x},\m u)}_2^2$ is solved on each machine for all $i \in\mbb{I}(g)$, thereby reducing the overall computational time.
	\end{itemize} 
	Note that the second strategy implies that the complexity of our approaches does \textit{not} necessarily depend on the problem dimension---instead, it relies on the number of different nonlinearities in the NDS. 
	At last, it is always possible to combine these strategies altogether to maximize the computational efficiency. 
	
	Next,  we discuss some recent results on Lipschitz continuity and QIB function sets. 
	It is generally known that, for $\m G = \m I$ and $\m f(\cdot)$ satisfies Lipschitz continuity condition with constant $\gamma_l$, then $\m f(\cdot)$ is also QIB with $\gamma_{q1} = \gamma_l^2$ and $\gamma_{q2} = 0$ \cite{Abbaszadeh2013}. 
	It is claimed in \cite{Abbaszadeh2013}, however, that the converse does not hold in general. The following statement is proven in that paper: if $\m f(\cdot)$ is OSL and QIB with $\gamma_{q2}$ positive, then $\m f(\cdot)$ is Lipschitz continuous.
	Our investigation reveals that if a function $\m f(\cdot)$ is QIB then it \emph{is necessarily} Lipschitz continuous. This finding is summarized in the following theorem.
	\vspace{-0.15cm}
	\begin{theorem}\label{thm:new_qib_lip}
		Suppose that $\m f(\cdot)$  is QIB with constants $\gamma_{q1},\gamma_{q2}\in\mbb{R}$ and $\m G = \m I$. Then, $\m f(\cdot)$ is also Lipschitz continuous where the constants $\gamma_{q1}, \gamma_{q2}$ necessarily satisfy the inequality
		\begin{align}
			2 \gamma_{q1}+|\gamma_{q2}|^2 \geq 0. \label{eq:qib_constants_condition}
		\end{align}
	\end{theorem}
	\vspace{-0.15cm}
	The above result (proved in \cite{Nugroho2020insights}) shows that QIB \textit{implies} Lipschitz continuity. Moreover since Lipschitz continuity implies QIB, it can be concluded that both class of  functions are the same. In fact, for given QIB constants $\gamma_{q1}, \gamma_{q2}$, the Lipschitz constant is computed as $\gamma_l = \sqrt{2\gamma_{q1} + \abs {\gamma_{q2}}^2}.$ Note that the condition \eqref{eq:qib_constants_condition} is necessary, i.e., it holds for every QIB functions, but \textit{not} sufficient.   
	In addition to this, our recent work \cite{Nugroho2020insights} also pinpoints some mistakes in the numerical section of \cite{Abbaszadeh2013,zhang2012full}, where we provide the corresponding correct conditions. 
	The next section focuses on the development of a simple interval-based algorithm to solve global maximization problems. }

\vspace{-0.2cm}
\section{A Derivative-Free Interval-Based Global Maximization via Branch-and-Bound Routines}\label{sec:interval_optimization}
After showing that the parameterization of NDS can actually be transformed into solving global maximization problems, our next step is finding a way to solve such problems.
In general there are two distinct approaches that can be pursued in solving global optimization problems: \textit{point-based method} and \textit{interval-based method} \cite{moa2007interval}. Note that the disadvantages of the former include the lack of optimality guarantees and the possibility of such methods converging to local maxima/minima \cite{moa2007interval}. In the context of NDS parameterization, if the solution from point-based optimization algorithms is suboptimal, then the resulting parameter is an \textit{under-approximation}. {For the interval-based algorithm, however, as it provides both upper and lower bounds towards the optimal solution, the resulting upper bound can be used to compute the corresponding parameter of the NDS. 
Based on the above considerations, we utilize an interval-based global optimization method for parameterizing the NDS. The following sequel succinctly discusses the interval-based global optimization method considered in this paper.}

\vspace{-0.2cm}
\subsection{Fundamentals of Interval Arithmetic and Interval Function  }\label{ssec:fia}
{
An {interval}  $\left[a,b\right]\subset \mbb{R}$ is formally defined as follows 
\begin{align}
\left[a,b\right] := \{x\in\mbb{R}\,|\,a\leq x \leq b\}.\nonumber
\end{align} 
In IA, it is common to represent a real variable $x$ by its interval so that $x\in\left[x\right] := \left[\barbelow{x},\bar{x}\right]$. It follows directly from its definition that $\barbelow{x} = \mathrm{inf}\left(\left[x\right]\right)$ and $\bar{x} = \mathrm{sup}\left(\left[x\right]\right)$. An interval $\left[x\right]$ is said to be \textit{degenerate} if it contains one element $x$ such that $x = \barbelow{x} = \bar{x}$. }
The midpoint (or {center}) and width of an interval $\left[x\right]$ are respectively defined as 
\begin{align*}
\mathrm{mid}\left(\left[x\right]\right) := \frac{1}{2}(\barbelow{x}+\bar{x}),\quad \mathrm{width}\left(\left[x\right]\right) := \bar{x}-\barbelow{x}.
\end{align*}
For brevity, we define $\abs{[x]} := \mathrm{width}\left(\left[x\right]\right)$. A binary operation of two intervals $X := \left[x\right]$ and $Y := \left[y\right]$ is defined as follow
\begin{align*}
X\diamond Y := \{z\in\mbb{R}\,|\,z = x\diamond y,\,x\in X,\,y\in Y\},
\end{align*}
where the notation $\diamond$ represents any binary operator from the set of elementary arithmetic operators $\{+,-,\times,\div\}$. Readers are referred to \cite{SCHULZEDARUP2018135,daumas2009verified} for more comprehensive IA operations. 

An $n$-orthotope $\mathbfcal{S} = \prod_{i\in\mbb{I}(n)} \mathcal{S}_i$ can be regarded as an \textit{element} of $n$-dimensional intervals $\mbb{IR}^n$, i.e., $\mathbfcal{S}\in\mbb{IR}^n$ and consequently $\mathcal{S}_i\in\mbb{IR}$. By definition, the notation $\m z\in \mathbfcal{S}$ means that $z_i\in\mathcal{S}_i$ for all $i\in\mbb{I}(n)$. For two $n$-orthotopes $\mathbfcal{S},\mathbfcal{S}'\in\mbb{IR}^n$, the notion $\mathbfcal{S}\subseteq\mathbfcal{S}'$ holds if for all $i\in\mbb{I}(n)$ we have $\mathcal{S}_i\subseteq\mathcal{S}'_i$ \cite{Moore2009}. 
The interval representation of a real-valued function $f(\cdot)$ is denoted by $f^{I}(\cdot)$ \cite{hansen2003global}. Such function is generally referred to as \textit{interval-valued function}. An interval function $f^{I}(\cdot)$, which is an interval representation of $f(\cdot)$, is termed as an \textit{interval extension} of $f(\cdot)$---see Definition \ref{def:interval_extension}. 
\vspace{-0.1cm}
\begin{mydef}\label{def:interval_extension}
	Let $f:\mbb{R}^n\rightarrow\mbb{R}$ be a mapping with $\mathbfcal{S}\subset \mbb{R}^n$ as the domain of interest. An interval function $f^I:\mbb{IR}^n\rightarrow\mbb{IR}$ is said to be an \textit{interval extension} of $f(\cdot)$ if, for all $\m z\in\mathbfcal{S}$, $f(\m z) = f^{I}\left([\m z, \m z]\right)$ where
	$[\m z, \m z] = \prod_{i\in\mbb{I}(n)} [z_i,z_i]$.
\end{mydef}
\vspace{-0.1cm}
{
The construction of $f^{I}(\cdot)$ can be done by simply substituting each occurrence of $x_i$ in $f(\cdot)$ with $\mathbfcal{S}_i$ for all $i\in\mbb{I}(n)$ \cite{moa2007interval}. At this point, we should be aware of one particular disadvantage of IA referred to as \textit{dependency effect}: the interval evaluation of $f(\cdot)$ depends on the expression of $f(\cdot)$ used to implement the computation \cite{moa2007interval}. 
With that in mind, in using IA for solving global optimization problems, one has to carefully choose the way of expressing the interval extension of the objective value in order to obtain the best results. }
The next definition presents one important property of an interval extension, namely \textit{inclusion isotonic}, which is useful for computing interval extensions. 
\vspace{-0.1cm}
\begin{mydef}\label{def:inclusion_isotonic}
	Let $f:\mbb{R}^n\rightarrow\mbb{R}$ be a mapping with $f^{I}:\mbb{IR}^n\rightarrow\mbb{IR}$ as its interval extension. For $\mathbfcal{S},\mathbfcal{S}'\in\mbb{IR}^n$ such that $\mathcal{S}_i\subseteq\mathcal{S}'_i$ for every $i\in\mbb{I}(n)$, then $f^{I}(\cdot)$ is \textit{inclusion isotonic} if and only if $f^{I}(\mathbfcal{S})\subseteq f^{I}(\mathbfcal{S})'$. 
\end{mydef}
\vspace{-0.1cm}
{
Herein, it is assumed that most computer programs represent real numbers using outward rounding \cite{van2010global}. This assumption allows every interval extension of any real function satisfies the inclusion isotonic property \cite{hansen2003global}. The fundamental theorem of interval arithmetic, presented below, guarantees that any interval extension of a real function satisfying the inclusion isotonicity encompasses the range of that function \cite{hansen2003global}. }
\vspace{-0.1cm}
\begin{mylem}[Fundamental Theorem of Interval Arithmetic]\label{lem:interval_arithmetic_fundamental_thrm}
	Let $f^{I}:\mbb{IR}^n\rightarrow\mbb{IR}$ be an interval extension of $f:\mbb{R}^n\rightarrow\mbb{R}$ which is inclusion isotonic. Then, for all $z_i\in\mathcal{S}_i$ where $i\in\mbb{I}(n)$ and $\mathbfcal{S}\in\mbb{IR}^n$, $f^{I}(\mathbfcal{S})$ contains the range of $f(\m z)$.
\end{mylem}  
\vspace{-0.1cm}
The proof of Lemma \ref{lem:interval_arithmetic_fundamental_thrm} can be obtained from \cite{hansen2003global}. In the above lemma, the range of $f(\cdot)$ for a given domain $\mathbfcal{S}$ is defined as $\mathrm{range}(f):= \left[\inf_{\m z \in\mathbfcal{S}} f(\m z),\sup_{\m z \in\mathbfcal{S}} f(\m z)\right]$. If we know the domain of $f(\cdot)$, then we simply refer the upper and lower bounds of the range of $f(\cdot)$ as $\sup f(\m z)$ and $\inf f(\m z)$---the same goes for expressing the upper and lower bounds of any interval-valued function. 
For the sake of simplicity, we write $f(\mathbfcal{S})$ to represent the range of $f(\cdot)$, that is, $f(\mathbfcal{S}) = \mathrm{range}(f)$. The fundamental theorem of arithmetic ensures that $f(\mathbfcal{S})\subseteq f^{I}(\mathbfcal{S})$ always holds for any valid $f^{I}(\cdot)$. 
An interval extension $f^I(\cdot)$ satisfying $f^{I}(\mathbfcal{S})= f(\mathbfcal{S})$ is said to be \textit{minimal} or \textit{exact}, and we say that its endpoints are \textit{sharp} \cite{moa2007interval,hansen2003global}. 
{ Readers are referred to \cite{hansen2003global,Moore2009} for more thorough basics on interval arithmetic, functions, and analysis.}
Next, we discuss the principle of solving global maximization problems based on IA.

\setlength{\textfloatsep}{5pt}
{\small \begin{algorithm}[t]
		\caption{\text{Interval-Based Algorithm for Unconstrained} \text{Global Maximization}}\label{alg:BnB-IA}
		\DontPrintSemicolon 
		\textbf{input:} $\mathbf{\Omega}$, $\epsilon_h$, $ \epsilon_{\Omega}$, $h(\cdot)$, $h^I(\cdot)$\;
		\textbf{initialize:} $\mathbfcal{S}_1 = \mathbf{\Omega}$, $C = \{\mathbfcal{S}_1\}$,  $\tilde{C} = \emptyset$\;
		\textbf{compute:} $\bar{\mathbfcal{S}} = \argmax_{\mathbfcal{S}_i\in C} \sup \left( h^I\hspace{-0.05cm}\left(\mathbfcal{S}_i\right) \right)$, 
		$u = \sup \hspace{-0.00cm}\left( h^I\hspace{-0.05cm}\hspace{-0.0cm}\left(\bar{\mathbfcal{S}}\right) \hspace{-0.00cm}\right)$, $l = h\left(\mathrm{mid}\left(\mathbfcal{S}_1\right)\right) \label{alg:l_1}$\;
		\While{$u-l > \epsilon_h$ {\bf and}   $\abs{\bar{\mathbfcal{S}}}> \epsilon_{\Omega}$\label{alg:first_loop}}{
			$	\{C,\bar{\mathbfcal{S}},l,u\}\leftarrow \text{OptBnB}\left({C,\bar{\mathbfcal{S}},l}\right)$
		}
		\If{$u-l > \epsilon_h$}{
			\ForEach{$\mathbfcal{S}_j\in C$ {\bf and} $\abs{\mathbfcal{S}_j}> \epsilon_{\Omega}$}{
				$\tilde{C}\leftarrow \tilde{C}\cup \{\mathbfcal{S}_j\}$\;
			}
		}
		\While{$u-l > \epsilon_h$ {\bf and}  $\tilde{C}\neq \emptyset$\label{alg:second_loop}}{
			$\barbelow{\mathbfcal{S}}\leftarrow \argmax_{\mathbfcal{S}_i\in \tilde{C}} \inf \left( h^I\hspace{-0.05cm}\left(\mathbfcal{S}_i\right) \right)$\;
			$	\{C,\bar{\mathbfcal{S}},l,u\}\leftarrow \text{OptBnB}\left({C,\barbelow{\mathbfcal{S}},l}\right)$, $\tilde{C}\leftarrow \emptyset$\;
			\ForEach{$\mathbfcal{S}_j\in C$ {\bf and} $\abs{\mathbfcal{S}_j}> \epsilon_{\Omega}$}{
				$\tilde{C}\leftarrow \tilde{C}\cup \{\mathbfcal{S}_j\}$\;
			}
		}
		
		\If{$u-l > \epsilon_h$}{
			\ForEach{$\mathbfcal{S}_j \in C$}{
				$l \leftarrow  \max\left(l, h\left(\inf\left(\mathbfcal{S}_j\right)\right), h\left(\sup\left(\mathbfcal{S}_j\right)\right)\right)\label{alg:l_2}$\;	
				\ForEach{$\mathbfcal{S}_k\in C$}{
					\If{$\sup \left(  h^I\hspace{-0.05cm}\left(\mathbfcal{S}_k \right)\right) < l$}{%
						$C\leftarrow C\setminus \{\mathbfcal{S}_k\}$\;
					}
				}
				\If{$u-l \leq \epsilon_h$}{
					\textbf{break}\;
				}
			}
		}
		\textbf{output:} $C$, $l$, $u$\;
	\end{algorithm}
}
\setlength{\floatsep}{5pt}
\vspace{-0.2cm}

{\small \begin{algorithm}[t]
		\caption{\textit{OptBnB}}\label{alg:opt_bnb}
		\DontPrintSemicolon
		\textbf{input:} $C$, $\mathbfcal{S}'$, $l$\;
		$C \leftarrow C\setminus \{\mathbfcal{S}'\}$, 
		$\mathbf{\Omega}' = \argmax_{\mathbf{\Omega}_i\in\mathbfcal{S}'} \abs{\mathbf{\Omega}_i}$\;
		$\mathbf{\Omega}'_l = \left[\inf \left(\mathbf{\Omega}'\right),\abs{\mathbf{\Omega}'}\right]$, $\mathbf{\Omega}'_r = \left[\abs{\mathbf{\Omega}'}, \sup \left(\mathbf{\Omega}'\right)\right]$\;
		$\mathbfcal{S}_l = \mathbf{\Omega}'_1\times\cdots\times \mathbf{\Omega}'_l\times\cdots\times\mathbf{\Omega}'_p$, $\mathbfcal{S}_r = \mathbf{\Omega}'_1\times\cdots\times \mathbf{\Omega}'_r\times\cdots\times\mathbf{\Omega}'_p$\;
		$C \leftarrow C\cup\{\mathbfcal{S}_l,\mathbfcal{S}_r\}$\;
		$\mathbfcal{S}_m = \argmax_{\mathbfcal{S}_i\in C} \inf \left( h^I\hspace{-0.05cm}\left(\mathbfcal{S}_i\right) \right)$\;
		$l \leftarrow  \max\left(l, h\left(\mathrm{mid}\left(\mathbfcal{S}_m\right)\right)\right)\label{alg:l_3}$\;
		\ForEach{$\mathbfcal{S}_i \in C$}{
			\If{$\sup \left(  h^I\hspace{-0.05cm}\left(\mathbfcal{S}_i \right)\right) < l$}{%
				$C\leftarrow C\setminus \{\mathbfcal{S}_i\}$\;
			}
		}
		$\mathbfcal{S}'' = \argmax_{\mathbfcal{S}_i\in C} \sup \left( h^I\hspace{-0.05cm}\left(\mathbfcal{S}_i\right) \right)$, $u = \sup \hspace{-0.00cm}\left( h^I\hspace{-0.05cm}\hspace{-0.0cm}\left(\mathbfcal{S}''\right) \hspace{-0.00cm}\right)$\;
		\textbf{output:} $C$, $\mathbfcal{S}''$, $l$, $u$\;
\end{algorithm}}

\vspace{-0.2cm}
\subsection{An Interval-Based BnB Algorithm for Global Maximization}\label{ssec:gmia}
We start this section by considering the following problem
\begin{align}
h := \max_{\m \omega \in \mathbf{\Omega}} \,h(\m \omega),\label{eq:max_basic_problem}
\end{align}  
where the function $h:\mbb{R}^p\rightarrow\mbb{R}$ denotes the objective function to be maximized inside a $p$-orthotope $\mathbf{\Omega}$ and $h$ denotes the objective function value. Since we do not have any constraint other than $\mathbf{\Omega}$, then Problem \eqref{eq:max_basic_problem} is categorized as \textit{unconstrained global maximization problem}. The above problem generalizes the way of computing the corresponding constants for each class of nonlinearity presented in Tab. \ref{tab:nonlinear_class}. 
{In this paper, the notion $\m \omega^*$ denotes any maximizer of \eqref{eq:max_basic_problem} and $h^* = h(\m \omega^*)$ denotes the optimal value of \eqref{eq:max_basic_problem}. For this problem to be meaningful, it is also assumed that $\m \omega^* \in \m \Omega$.}

{ The term \textit{interval-based algorithm} usually refers to any algorithm that utilizes IA for solving global optimization problem.
IA has been widely utilized over past decades for solving global optimization problems, either constrained or unconstrained---see \cite{van2010global,ichida1979interval,moa2007interval,hansen2003global,Moore2009,Hansen1980,skelboe1974computation,jaulin2001applied}.
Global maximization problems, in general, can be categorized into \textit{fathoming} and \textit{localization} problems \cite{moa2007interval}. The former focuses on finding the value of the global maximum where the latter concerns with identifying the location of the global maximum. Indeed, various interval-based algorithms for global optimization have been developed since several past decades---for some references, see \cite{van2010global,ichida1979interval,moa2007interval,hansen2003global,Moore2009,skelboe1974computation,Hansen1980,jaulin2001applied}. However, as pointed out in \cite{moa2007interval}, these algorithms suffer from several drawbacks: \textit{(i)} they combine the fathoming problem with localization problem, \textit{(ii)} they do not provide any insight on the optimality of the computed solutions, \textit{(iii)} for localization, they do not consider inner-approximation (finding hyperrectangles contained within a specified distance $\delta > 0$) and outer-approximation (finding hyperrectangles covering that distance).}


{
In addition to the aforementioned limitations, the currently available interval-based global optimization tools utilize sophisticated methods from point-based optimization with Newton's method \cite{Moore2009,hansen2003global,Hansen1980} being the most prevalent. Note that Newton's method requires the computation of Jacobian matrix as well as the Hessian (see \cite{Hansen1980}) which requires the objective function $h(\cdot)$ to be belonged to $\mathcal{C}^2$. To have $h(\cdot)\in \mathcal{C}^2$, then it is necessary that $\m f(\cdot)\in \mathcal{C}^3$, which potentially restricts the applicability of our method to parameterize NDS.
With that in mind, since it is assumed that $\m f(\cdot)\in \mathcal{C}^1$ as presented in Assumption \ref{asmp:1}, we decide to utilize a derivative-free interval-based global optimization algorithm developed via a branch-and-bound (BnB) framework. In particular, we adopt the Algorithm MS3 developed in \cite{moa2007interval}, which is an improved variation of the Moore-Skelboe algorithm \cite{skelboe1974computation,Moore1976} with stopping criteria specified by some thresholds---see Algorithm \ref{alg:BnB-IA}. Readers are referred to \cite{moa2007interval} for the detailed advantages of this algorithm relative to other contemporaries for solving global optimization problems. The following definition describes the concept of \textit{cover} \cite{van2004termination} that is pivotal for the algorithm.}
%

\vspace{-0.15cm}
\begin{mydef}\label{def:cover}
Let $\mathbf{\Omega}$ be a set. If $C = \{\mathbfcal{S}_i\}^N_{i = 1}$ is a sequence of $N$ nonempty subsets of $\mathbf{\Omega}$, then we say that $C$ is a \textit{cover} in $\mathbf{\Omega}$. Specifically, if the following holds
\begin{align*}
\mathbf{\Omega} = \bigcup_{i\in\mbb{I}(N)} \mathbfcal{S}_i,
\end{align*}  
then we say that $C$ is a \textit{cover} for $\mathbf{\Omega}$.
\end{mydef}   
\vspace{-0.15cm}
In the above definition, $N:=\mathrm{card}(C)$, which denotes the number of elements in $C$. 
The objective of the interval-based algorithm is to find the best (\textit{smallest}) interval $[l,u]$ containing $h^*$. If this interval is degenerate, then we have successfully found $h^*$. Otherwise in the case when $l<u$, we are assured that the optimal solution cannot be greater than $u$ and smaller than $l$. This can be achieved by iteratively decreasing the upper bound $u$ and increasing the lower bound $l$.
The reduction of $u$ can be performed by splitting a subset $\mathbfcal{S}\in C$ having the biggest upper bound of interval evaluation in $C$. Likewise, increasing $l$ can be done by splitting subset $\mathbfcal{S}\in C$ having the smallest lower bound in $C$.  To that end, we present an interval-based algorithm---given in Algorithm \ref{alg:BnB-IA}---for solving \eqref{eq:max_basic_problem}. 
In this algorithm, it is necessary to define two positive constants $\epsilon_h$ and $\epsilon_{\Omega}$: the constant $\epsilon_h$ is useful to bound the interval containing the optimal solution, whereas $\epsilon_{\Omega}$ is useful to provide a lower bound towards the width of all subsets $\mathbfcal{S}_i$ in the cover---if the width of a subset $\mathbfcal{S}_i$ is smaller than $\epsilon_{\Omega}$, then we assume that $\mathbfcal{S}_i$ cannot be split into two smaller subsets. The \textit{width} of a subset $\mathbfcal{S}_i$ is defined as 
$\abs{\mathbfcal{S}_i} := \max_{{\Omega}_j\in \mathbfcal{S}_i} \abs{{\Omega}_j}.$
%
The algorithm is said to be successful in computing the best solution when it holds that $u-l\leq \epsilon_h$---in this case, we say that $h^*$ is $\epsilon_h-$\textit{optimal}. Otherwise, $u-l$ is the smallest interval containing $h^*$ that can be obtained. 
{ After Algorithm \ref{alg:BnB-IA} terminates, we can use $u$---the upper bound---to compute the corresponding constants for function sets listed in Tab. \ref{tab:nonlinear_class}; see Section \ref{sec:numerical_tests}.}
 {\small \begin{algorithm}[t]
 		\caption{\text{Computing Lipschitz Constants $\gamma_{l_1}$ and $\gamma_{l_2}$}}\label{alg:lip_comp}
 		\DontPrintSemicolon
 		{
 		\textbf{input:} $\mathbf{\Omega}$, $\epsilon_h$, $\epsilon_{\Omega}$, $\m f(\cdot)$, \textit{case} \;
 		\Switch{\textit{case}}{
 			\Case{$1$}{
 				\textbf{define:} $h(\m x):=\sum_{i\in \mbb{I}(n)}\norm{\nabla_{\hspace{-0.05cm}x} f_i(\m{x})}_2^2$ and construct $h^I(\m x)$\;
 				\textbf{obtain:} $u$ and $l$ from Algorithm \ref{alg:BnB-IA} given $\mathbf{\Omega}$, $\epsilon_h$, $ \epsilon_{\Omega}$, $h(\cdot)$, $h^I(\cdot)$ \label{alg:step_lip_1}\;
 				\textbf{output:} $\gamma_{l_1} = \sqrt{u}$\;
 			}
 			\Case{$2$}{
 				\ForEach{$z \in\{a,b,c,d,e\}$ \label{alg3:step8}}{
 					\textbf{define:} $h_z(\m x):=\norm{\nabla_{\hspace{-0.05cm}x} f_z(\m{x})}_2^2$ and construct $h^I_z(\m x)$ as well as $\mathbf{\Omega}_z$\;
 					\textbf{obtain:} $u_z$ and and $l_z$ from Algorithm \ref{alg:BnB-IA} given $\mathbf{\Omega}_z$, $\epsilon_h$, $ \epsilon_{\Omega}$, $h_z(\cdot)$, $h^I_z(\cdot)$\label{alg:step_lip_2}\;
 				}
 				\textbf{output:} $\gamma_{l_2} = \sqrt{c_a u_a+c_b u_b+c_c u_c+c_d u_d+c_e u_e}$\;
 			}
 		}}
 \end{algorithm}}
 It is important to note that this algorithm is a modification of Algorithm MS3 proposed in \cite{moa2007interval} with two main differences. Firstly, Algorithm \ref{alg:BnB-IA} is specifically designed for solving global maximization problem \eqref{eq:max_basic_problem} and presented using rigorous mathematical notations for the ease of analysis and implementation, and secondly, we have added an \textit{explicit} additional step to ensure that the splitting of a subset $\mathbfcal{S}\in C$ is possible only when the width of $\mathbfcal{S}$ is greater than $\epsilon_{\Omega}$. {The next theorem---see Appendix \ref{appdx:B} for the proof---shows that Algorithm \ref{alg:BnB-IA} computes the best interval for $h^*$.}
 \vspace{-0.1cm}
\begin{theorem}\label{thm:IA_based_algorithm}
When Algorithm \ref{alg:BnB-IA} terminates, then 
\begin{enumerate}
	\item $h^*\in\left[l,u\right]$,  where $\left[l,u\right]$ is the best interval for $h^*$,
	\item $\m \omega^*\in C = \prod_{i\in\mbb{I}(N)}\mathbfcal{S}_i$ where $C$ is a cover in $\m\Omega$ at the end of the algorithm.
\end{enumerate}
\end{theorem}
\vspace{-0.1cm}


Since the lower bound $l$ is not required for computing NDS' parameters, basically one can stop Algorithm \ref{alg:BnB-IA} when the first loop in Step \ref{alg:first_loop} terminates. However, by obtaining the best lower bound $l$, it gives us a sense on the location of $h^*$ as well as an optimality certificate through the definition of $\epsilon_h-$optimal. 

\vspace{-0.2cm}
\begin{myrem}\label{rem:improve_l}
	{
	The lower bound $l$ in Algorithm \ref{alg:BnB-IA} is determined via the evaluation of objective function using the midpoint of the corresponding subset, say, $\mathbfcal{S}_i$. As discussed in \cite{moa2007interval}, one can also use the evaluation of either corners of $\mathbfcal{S}_i$. 
	Other than these approaches, deterministic point-based techniques can also be employed such as gradient ascent, Newton method, conjugate gradient method, and interior-point method, as well as stochastic ones such as quasi-random algorithms and heuristics. }
\end{myrem}
\vspace{-0.2cm}

\vspace{-0.2cm}

\section{Numerical Examples}\label{sec:numerical_tests}
In this section we showcase the proposed methodologies and algorithms for parameterizing some NDS that include traffic networks, synchronous generator, and the motion of a moving object into the following function sets: Lipschitz continuous, OSL, and QIB. { All simulations are performed using {MATLAB} R2020a running on a 64-bit Windows 10 with 3.4GHz Intel\textsuperscript{R} Core\textsuperscript{TM} i7-6700 CPU and 16 GB of RAM.
YALMIP's \cite{Lofberg2004} optimization package together with MOSEK's \cite{andersen2000mosek} SDP solver
are used to solve any semidefinite programs. All MATLAB codes associated in this section are available on GitHub \cite{Nugroho2020github}.}

\setlength{\textfloatsep}{10pt}
\begin{table*}
	\scriptsize
	\vspace{-0.1cm}
	\centering 
	\caption{{Numerical results for Lipschitz constants computation on traffic networks. Notations:  $\gamma_{l_1}$ and  $\gamma_{l_2}$ are Lipschitz constants, $\barbelow{\gamma}_{l_1}$ and $\barbelow{\gamma}_{l_2}$ are the corresponding lower bounds, \textit{gap} refers to {optimality gap}, which essentially computes the difference between $u$ and $l$, $\epsilon_h$-\textit{optimal} is checked when the solution is $\epsilon_h-$optimal, and $\Delta t$ is the total computation time.}}
	\label{tab:traffic_simul_data}
	\vspace{-0.2cm}
	\renewcommand{\arraystretch}{1.5}{
	\begin{tabular}{|c|c|c|c|c|c|c|c|c|c|c|c|}
		\hline
		\multirow{2}{*}{$n$} & \multicolumn{5}{c|}{Case 1} & \multicolumn{5}{c|}{Case 2} & \multicolumn{1}{c|}{Analytical \cite{nugroho2018journal}}\\ \cline{2-12} 
		& $\gamma_{l_1}$  &  $\barbelow{\gamma}_{l_1}$ & \textit{gap} $(\times 10^{-5})$  & \textit{$\epsilon_h$-optimal?} & $\Delta t$ (s) &   $\gamma_{l_2}$  &  $\barbelow{\gamma}_{l_2}$ & \textit{gap} $(\times 10^{-5})$  & \textit{$\epsilon_h$-optimal?} & $\Delta t$ (s) & $\gamma_l$  \\ \hline\hline
		$31$	& $0.4579$ &   $0.4578$ &  $9.8531$ & \cmark & $  1.24$  & $0.4579$ & $0.4552$ & $7.6168$ &  \cmark & $0.0330$ & $0.5069$ \\ \hline
		$61$	& $0.6445$ &   $0.6444$ &  $9.9494$ & \cmark & $  4.81$  & $0.6445$ & $0.6408$ & $7.6168$ &  \cmark & $0.0345$ & $0.7141$ \\ \hline
		$91$	& $0.7881$ &   $0.7881$ &  $9.9974$ & \cmark & $ 11.13$  & $0.7881$ & $0.7836$ & $7.6168$ &  \cmark & $0.0266$ & $0.8735$ \\ \hline
		$121$	& $0.9093$ &   $0.9093$ &  $9.9975$ & \cmark & $ 20.55$  & $0.9093$ & $0.9041$ & $7.6168$ &  \cmark & $0.0258$ & $1.0080$ \\ \hline
		$151$	& $1.0162$ &   $1.0161$ &  $9.9736$ & \cmark & $ 33.46$  & $1.0162$ & $1.0103$ & $7.6168$ &  \cmark & $0.0249$ & $1.1265$ \\ \hline
		$181$	& $1.1128$ &   $1.1128$ &  $9.9976$ & \cmark & $ 49.47$  & $1.1128$ & $1.1064$ & $7.6168$ &  \cmark & $0.0188$ & $1.2338$ \\ \hline
		$211$	& $1.2017$ &   $1.2017$ &  $9.9976$ & \cmark & $ 68.99$  & $1.2017$ & $1.1948$ & $7.6168$ &  \cmark & $0.0181$ & $1.3324$ \\ \hline
		$241$	& $1.2844$ &   $1.2844$ &  $9.9977$ & \cmark & $ 93.13$  & $1.2844$ & $1.2770$ & $7.6168$ &  \cmark & $0.0277$ & $1.4242$ \\ \hline
		$271$	& $1.3622$ &   $1.3621$ &  $9.9857$ & \cmark & $120.51$  & $1.3622$ & $1.3543$ & $7.6168$ &  \cmark & $0.0232$ & $1.5104$ \\ \hline
		$301$	& $1.4357$ &   $1.4357$ &  $9.9977$ & \cmark & $151.54$  & $1.4357$ & $1.4274$ & $7.6168$ &  \cmark & $0.0213$ & $1.5919$ \\ \hline
	\end{tabular}}
	\vspace{-0.1cm}
\end{table*}
\setlength{\floatsep}{10pt}

\setlength{\textfloatsep}{10pt}
\begin{table*}
	\scriptsize
	\vspace{-0.1cm}
	\centering 
	\caption{{Lipschitz constants of highway traffic network computed with different methods: interval-based algorithm (IA), point-based global search (GS), and point-based quasi-random search with Sobol (S) and Halton (H) sequences.}}
	\label{tab:traffic_simul_data_comparison}
	\vspace{-0.2cm}
	\renewcommand{\arraystretch}{1.5}{
		\begin{tabular}{|c|c|c|c|c|c|c|c|c|c|c|c|}
			\hline
			\multirow{2}{*}{$n$} & \multicolumn{5}{c|}{Case 1} & \multicolumn{5}{c|}{Case 2} & \multicolumn{1}{c|}{$\max_{\m x\in\mathbf{\Omega}} \norm{\mathrm{D}_x \m f(\m x)}_2$}\\ \cline{2-12} 
			& $\gamma_{l_1}$(IA)  &  $\barbelow{\gamma}_{l_1}$(IA) & $\gamma_{l_1}$(GS)  & $\gamma_{l_1}$(S) & $\gamma_{l_1}$(H) &   $\gamma_{l_2}$(IA)  &  $\barbelow{\gamma}_{l_2}$(IA) & $\gamma_{l_2}$(GS)  & $\gamma_{l_2}$(S) & $\gamma_{l_2}$(H) &  $\gamma_l\,$(GS)  \\ \hline\hline
	$31$	& $0.4579$ &   $0.4578$ &  $0.4579$ & $0.3426$ & $0.3396$  & $0.4579$ & $0.4552$ & $0.4579$ &  $0.4523$ & $0.4497$  &  $0.1328$  \\ \hline
$61$	& $0.6445$ &   $0.6444$ &  $0.6445$ & $0.4556$ & $0.4435$  & $0.6445$ & $0.6408$ & $0.6445$ &  $0.6365$ & $0.6328$  &  $0.1329$  \\ \hline
$91$	& $0.7881$ &   $0.7881$ &  $0.7881$ & $0.5314$ & $0.5196$  & $0.7881$ & $0.7836$ & $0.7881$ &  $0.7783$ & $0.7738$  &  $0.1330$  \\ \hline
$121$	& $0.9093$ &   $0.9093$ &  $0.9093$ & $0.6054$ & $0.5832$  & $0.9093$ & $0.9041$ & $0.9093$ &  $0.8980$ & $0.8927$  &  $0.1330$  \\ \hline
$151$	& $1.0162$ &   $1.0161$ &  $1.0162$ & $0.6686$ & $0.6373$  & $1.0162$ & $1.0103$ & $1.0162$ &  $1.0035$ & $0.9976$  &  $0.1330$  \\ \hline
$181$	& $1.1128$ &   $1.1128$ &  $1.1128$ & $0.7242$ & $0.6964$  & $1.1128$ & $1.1064$ & $1.1128$ &  $1.0990$ & $1.0925$  &  $0.1330$  \\ \hline
$211$	& $1.2017$ &   $1.2017$ &  $1.2017$ & $0.7791$ & $0.7524$  & $1.2017$ & $1.1948$ & $1.2017$ &  $1.1867$ & $1.1797$  &  $0.1330$  \\ \hline
$241$	& $1.2844$ &   $1.2844$ &  $1.2844$ & $0.8239$ & $0.8086$  & $1.2844$ & $1.2770$ & $1.2844$ &  $1.2684$ & $1.2610$  &  $0.1330$  \\ \hline
$271$	& $1.3622$ &   $1.3621$ &  $1.3622$ & $0.8713$ & $0.8540$  & $1.3622$ & $1.3543$ & $1.3622$ &  $1.3452$ & $1.3373$  &  $0.1330$  \\ \hline
$301$	& $1.4357$ &   $1.4357$ &  $1.4357$ & $0.9083$ & $0.9019$  & $1.4357$ & $1.4274$ & $1.4357$ &  $1.4178$ & $1.4095$  &  $0.1330$  \\ \hline 
	\end{tabular}}
	\vspace{-0.4cm}
\end{table*}
\setlength{\floatsep}{10pt}

\begin{figure}[t]
	\vspace{-0.05cm}
	\centering 	\includegraphics[scale=0.29]{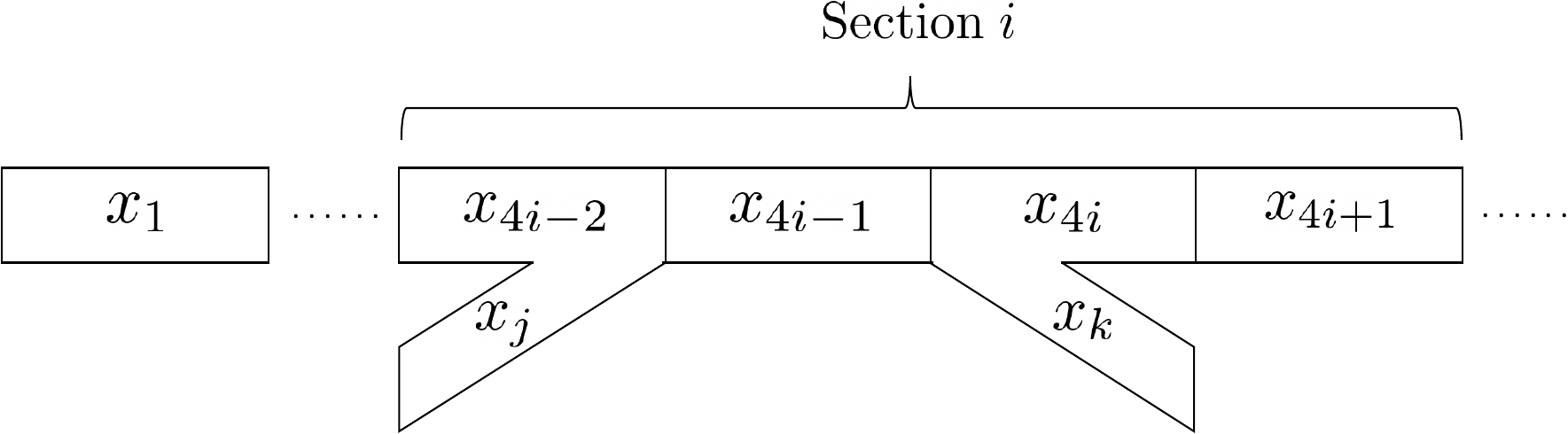}
	\vspace{-0.25cm}
	\caption{{The structure of the highway for traffic networks parameterization. Each Section $i$ for $i = 1,2,\hdots,s$ consists of four mainline segments, one on-ramp, and one off-ramp.}}
	\label{fig:traffic_structure}
	\vspace*{-0.05cm}
\end{figure}

\begin{figure}[t]
	\vspace{-0.1cm}
	\centering 	\includegraphics[scale=0.80]{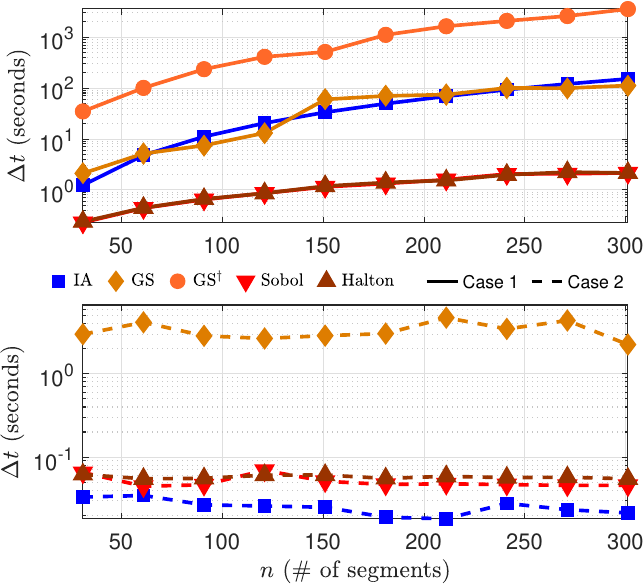}
	\vspace{-0.2cm}
	\caption{{Total computation time on calculating Lipschitz constants for the two cases with interval-based algorithm (IA), MATLAB's global search (GS), and stochastic, quasi-random search with Sobol and Halton sequences. The results from MATLAB's global search indicated by GS$^{\dagger}$ are obtained from using the 2-norm of Jacobian matrix \eqref{eq:lip_jacobian}.}}
	\label{fig:lip_traffic_comp_time}
	\vspace*{-0.2cm}
\end{figure}

\vspace{-0.2cm}
\subsection{Lipschitz Parameterization for Traffic Networks}\label{ssec:trans}
  { The dynamic model of traffic networks, particularly on a stretched highway consisting of on- and off-ramps, can be modeled by partitioning it into several segments of equal length $l$ such that the dynamics of each highway segment can be expressed in the form of ordinary differential equations \cite{nugroho2018}
	\begin{align*}
		\dot{\rho}_i (t)  
		&=  \frac{1}{l}\left(\sum_{j\in \mc{Q}^{\mr{in}}_i} q_{j}(t)-\sum_{j\in \mc{Q}^{\mr{out}}_i} q_{j} (t)\right),
	\end{align*}
	where $\mc{Q}^{\mr{in}}_i$ and $\mc{Q}^{\mr{out}}_i$ represent the sets of all incoming and outgoing traffic flows $q_{j}$ for segment $i$, and ${\rho}_i$ represents traffic density on that segment.} The above model assumes that the highway is uncongested---in this situation, the traffic density satisfies $\rho_i\in[0,\rho_c]$ where $\rho_c$ denotes the critical density. 
By employing Greenshield's fundamental diagram to represent the relation between traffic density and traffic flow \cite{greenshields1935study}, the nonlinearity of the overall dynamic model for this particular system ({ supposing that no segment is connected to both on- and off-ramps}) can be distinguished into five different types of nonlinearities, detailed as follows \cite{nugroho2018,nugroho2018journal}
\begin{subequations}\label{eq:traffic_nonlinearities}
	\begin{enumerate}[label=$\alph*$)]
		\item $f_a(\m x) = \delta x^2_i$ \hfill\refstepcounter{equation}(\theequation)
		\item $f_b(\m x) = \delta\left(x^2_i-x^2_{i-1}-x^2_j\right)$ \hfill\refstepcounter{equation}(\theequation)
		\item $f_c(\m x) = \delta\left(x^2_i-x^2_{i-1}\right)$ \hfill\refstepcounter{equation}(\theequation)\label{eq:traffic_nonlinearities_c}
		\item $f_d(\m x) = \delta\left(x^2_i-x^2_{i-1}+\alpha_jx^2_j\right)$ \hfill\refstepcounter{equation}(\theequation)
		\item $f_e(\m x) = -\delta \alpha_ix^2_i$, \hfill\refstepcounter{equation}(\theequation)
	\end{enumerate}	
\end{subequations}
where $x_i := \rho_i$ and $\delta := \frac{v_f}{l\rho_m}$; $v_f$ and $\rho_m$ denote the free flow speed and maximum density (maximum density $\rho_m$ is equal to $2\rho_c$). In this case, the following values, which are adapted from \cite{Contreras2016}, are considered: $v_f = 31.3$ m/s, $\rho_m = 0.053$ vehicles/m, and $l = 500$ m. The splitting ratio for off-ramps are chosen to be $\alpha_i = 0.5$.   {The highway consists of a certain structure, depicted in Fig. \ref{fig:traffic_structure}, such that it has $s$ sections and is containing $n$ number of segments in total (including ramps) where $n = 6s+1$. } 

\begin{figure*}[h]
	\centering
	  
	\hspace{-0.2cm}
	\subfloat[]{\includegraphics[scale=0.465]{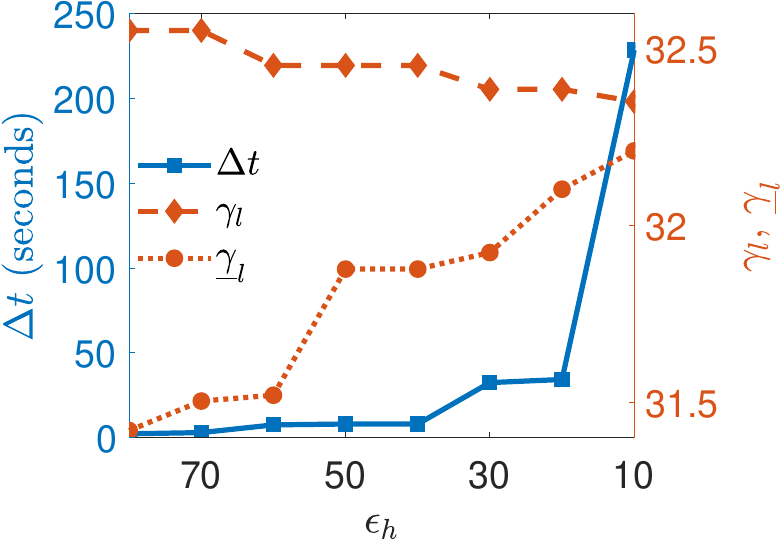}\label{Fig:gen_1a}}{}\hspace{-0.1em}
	\subfloat[]{\includegraphics[scale=0.465]{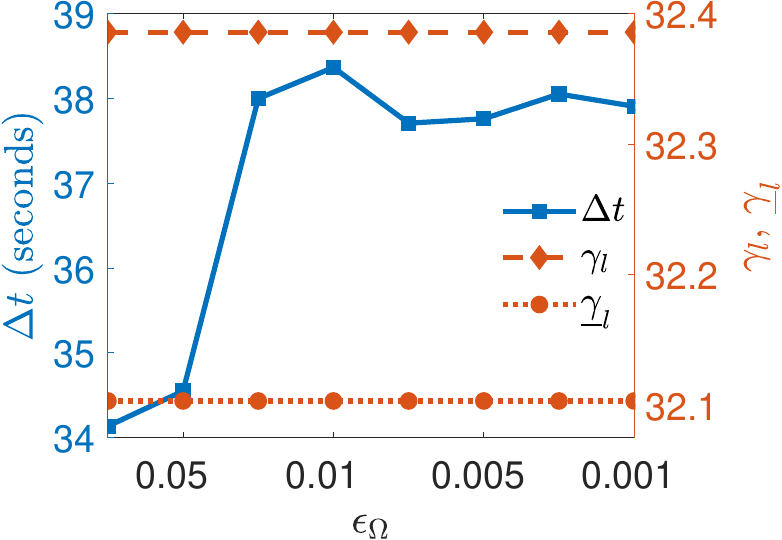}\label{Fig:gen_1b}}{}\hspace{-0.1em}
	\subfloat[]{\includegraphics[scale=0.465]{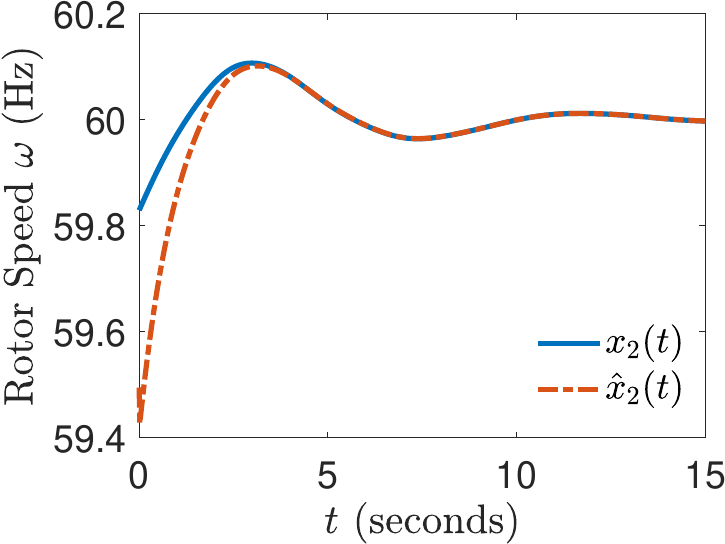}\label{Fig:gen_1c}}{\vspace{-0.1cm}}\hspace{-0.0em}
	\vspace{-0.05cm}
	\caption{{Numerical test results for synchronous generator: \textit{(a)} computational time $(\Delta t)$ and Lipschitz constant $(\gamma)$ for various $\epsilon_h$, \textit{(b)} the corresponding results for various $\epsilon_{\Omega}$, and \textit{(c)} the trajectories of estimated and actual generator rotor speed $(\omega)$ using a Lipschitz-based observer \cite{Phanomchoeng2010} with $\gamma_l = 32.3514$.}}
	\label{Fig:gen_1}	
	\vspace{-0.4cm}
\end{figure*} 

{
Herein, two cases are considered which depend on the way Lipschitz constants are computed: \textit{(i)} \textit{Case 1} which uses \eqref{eq:gamma_Lipschitz} provided in Theorem \ref{prs:Lipschitz_less_cons}, and \textit{(ii)} \textit{Case 2} which uses \eqref{eq:gamma_Lipschitz_2} provided in Corollary \ref{cor:Lipschitz}.
Specific for Case 2, firstly,
the corresponding set $\m \Omega$ is reduced so that the variables considered in the computation of $\max_{\m{x}\in \mathbf{\Omega}}\norm{\nabla_{\hspace{-0.05cm}x} f_z(\m{x})}_2^2$ for each $z\in \{a,b,c,d,e\}$ in \eqref{eq:traffic_nonlinearities} are only the ones that appear in the expression of \eqref{eq:traffic_nonlinearities}. Secondly, we only solve five different maximization problems corresponding to each type of nonlinearity in \eqref{eq:traffic_nonlinearities}. Both of these steps apply the approaches discussed in Section \ref{ssec:lip_qib_relation}. The expressions of $\norm{\nabla_{\hspace{-0.05cm}x} f_z(\m{x})}_2^2$ for each $z$ are listed as follows
\begin{subequations}\label{eq:traffic_nonlinearities_derivatives}
	\begin{enumerate}[label=$\alph*$)]
		\item $\norm{\nabla_{\hspace{-0.05cm}x} f_a(\m{x})}_2^2 = 4\delta^2 x^2_i$ \hfill\refstepcounter{equation}(\theequation)
		\item $\norm{\nabla_{\hspace{-0.05cm}x} f_b(\m{x})}_2^2 = 4\delta^2\left(x^2_i+x^2_{i-1}+x^2_j\right)$ \hfill\refstepcounter{equation}(\theequation)
		\item $\norm{\nabla_{\hspace{-0.05cm}x} f_c(\m{x})}_2^2 = 4\delta^2\left(x^2_i+x^2_{i-1}\right)$ \hfill\refstepcounter{equation}(\theequation)\label{eq:traffic_nonlinearities_c_der}
		\item $\norm{\nabla_{\hspace{-0.05cm}x} f_d(\m{x})}_2^2 = 4\delta^2\left(x^2_i+x^2_{i-1}+\alpha_j^2x^2_j\right)$ \hfill\refstepcounter{equation}(\theequation)
		\item $\norm{\nabla_{\hspace{-0.05cm}x} f_e(\m{x})}_2^2 = 4\delta^2 \alpha_i^2x^2_i$. \hfill\refstepcounter{equation}(\theequation)
	\end{enumerate}	
\end{subequations}
  {Due to the particular structure of the traffic network illustrated in Fig. \ref{fig:traffic_structure}, then for every $z\in \{a,b,c,d,e\}$, there exist $c_z\in \mbb{N}$ highway segments which nonlinearities are of type $z$.} 
Algorithm \ref{alg:lip_comp} presents the steps to compute Lipschitz constants for both cases---note that $\gamma_{l_1}$ corresponds to Case 1 whereas $\gamma_{l_2}$ corresponds to Case 2.
In this test, we compare different number of total segments $n = 31,61,\hdots,301$. 
Constants $\epsilon_h$ and $\epsilon_{\Omega}$ for Algorithm \ref{alg:lip_comp} are set to be $1\times 10^{-4}$ and $1\times 10^{-7}$ respectively whereas the region of interest is set to be $\m \Omega = [0,\rho_c]^n$ since the traffic is assumed to be uncongested.} 

{
  
After running Algorithm \ref{alg:lip_comp} using the aforementioned settings, we obtain results that are detailed in Tab. \ref{tab:traffic_simul_data}. All of the computed Lipschitz constants $\gamma_{l_1}$ and $\gamma_{l_2}$ are $\epsilon_h-$optimal, since the gap between $u$ and $l$ are all less than the predefined $\epsilon_h$. 
Interestingly, notice that the values of $\gamma_{l_1}$ and $\gamma_{l_2}$ are the same for both cases. This demonstrates that the formulations posed in \eqref{eq:gamma_Lipschitz} and \eqref{eq:gamma_Lipschitz_2} can indeed return similar Lipschitz constants with a fine precision. For both cases, the resulting Lipschitz constants also considerably smaller than analytical results, which can be advantageous in observer/controller design.
It is worth noting that, despite giving the same constants, Case 1 requires a significantly higher computation time than Case 2. This observation is expected since in Case 2, only five small optimization problems are solved and as such, the problem's complexity does \textit{not} depend on the dimension of $\m f(\cdot)$.    

Next, we compare the interval-based algorithm (Algorithm \ref{alg:lip_comp}) with stochastic point-based algorithm based on quasi-random \textit{low discrepancy sequences} (LDS). This approach is utilized in \cite{Nugroho2019} for Lipschitz constant approximation of synchronous generator model. We apply Algorithm 1 from \cite{Nugroho2019} with $10^4$ points from Sobol and Halton sequences generated through MATLAB's \texttt{sobolset} and \texttt{haltonset} functions. 
In addition, we also implement MATLAB's global search function \texttt{GlobalSearch} along with \texttt{fmincon} to find Lipschitz constants. 
The results of this comparison are provided in Tab. \ref{tab:traffic_simul_data_comparison} and Fig. \ref{fig:lip_traffic_comp_time}, where we also compare the results from using \eqref{eq:lip_jacobian} with MATLAB's global search.
It can be seen from this table that, for both cases, the Lipschitz constants computed via MATLAB's global search match with the ones from interval-based algorithm, while in contrast, the results from quasi-random search unsurprisingly are considerably lower for Case 1 but not so much for Case 2 (this numerically confirms that \eqref{eq:gamma_Lipschitz} is indeed less conservative than \eqref{eq:gamma_Lipschitz_2}). Notice the increasing gaps of Lipschitz constants computed from utilizing interval-based algorithm and quasi-random search in Case 1 as $n$ increases, which in turn suggests that the Lipshitz under-approximations computed via quasi-random search are getting worse as the problem's dimension increases. Lipschitz constant computation using \eqref{eq:lip_jacobian} with MATLAB's global search returns values that are relatively constant, which is $0.1330$. Despite giving the smallest Lipschitz constants, it is difficult to justify whether this value satisfies \eqref{eq:locally_Lipschitz} since MATLAB's global search does not guarantee optimality.
On the matter of computational time however (see Fig. \ref{fig:lip_traffic_comp_time}), particularly for Case 1, both interval-based algorithm and MATLAB's global search require much higher computation time compared to using LDS, which are much more scalable considering that $10^4$ points from the LDS are used. Nonetheless, for Case 2, our interval-based algorithm outperforms the other methods with the lowest computation time whereas MATLAB's global search requires much higher computation time. Note that MATLAB's global search with \eqref{eq:lip_jacobian} is the \textit{least} scalable compared to other methods in both cases. }


\begin{figure*}[ht]
	\centering
	\vspace{-0.2cm}
	\hspace{-0.62cm}
	\subfloat[]{\includegraphics[scale=0.47]{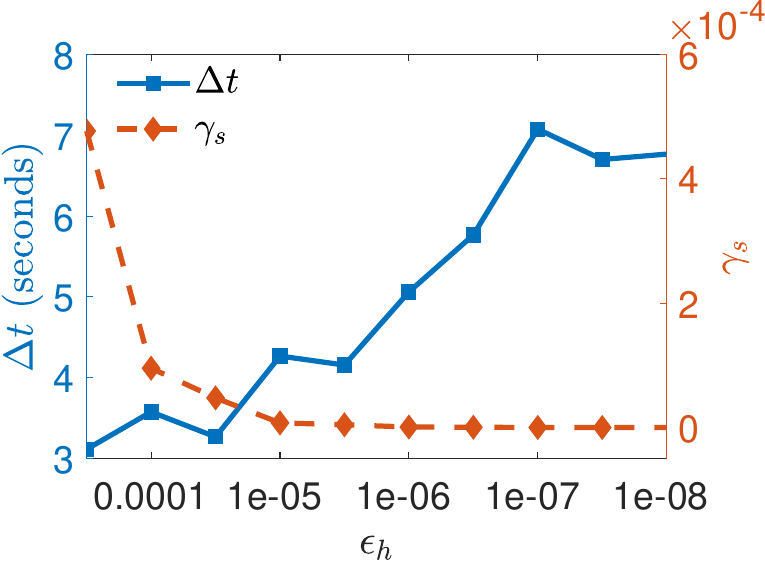}\label{Fig:mov_obj_1a}}{}\hspace{-0.1em}
	\subfloat[]{\includegraphics[scale=0.47]{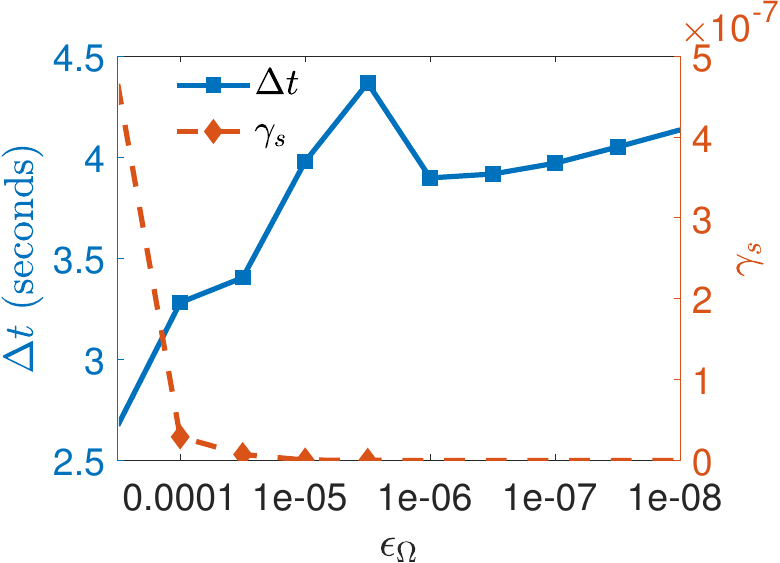}\label{Fig:mov_obj_1b}}{}\hspace{-0.1em}
	\subfloat[]{\includegraphics[scale=0.47]{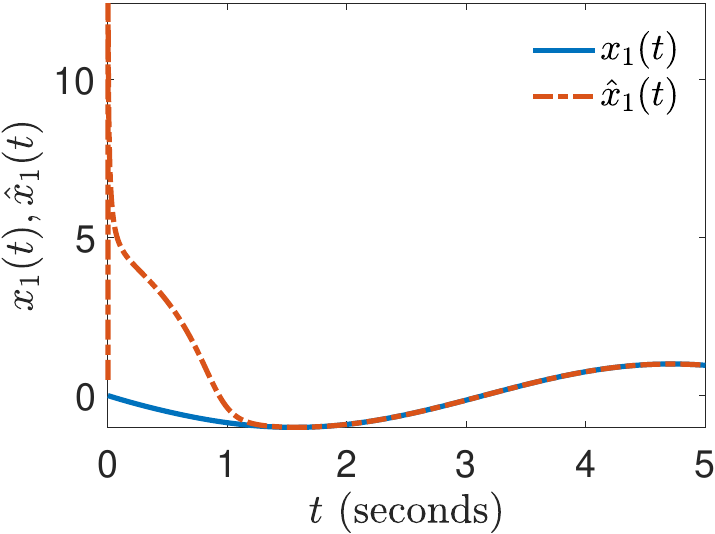}\label{Fig:mov_obj_1c}}{\vspace{-0.1cm}}\hspace{-0.0em}
	\vspace{-0.02cm}
	\caption{Numerical test results for dynamics of a moving object: \textit{(a)} computational time $(\Delta t)$ and OSL constant $(\gamma_s)$ for various $\epsilon_h$, \textit{(b)} the corresponding results for various $\epsilon_{\Omega}$, and \textit{(c)} the trajectories of actual $x_1(t)$ and estimated $\hat{x}_1(t)$ state with $\gamma_s = 0$, $\gamma_{q1} = 25015$, and $\gamma_{q2} = -9999.89$.}
	\label{Fig:mov_obj}	
	\vspace{-0.5cm}
\end{figure*}

\setlength{\textfloatsep}{10pt}
\begin{table}[t]
	\centering
	\vspace{-0.1cm}
	\caption{{Comparison of Lipschitz constants for synchronous generator model computed via interval-based algorithm (IA), point-based global search (GS), and point-based quasi-random search with Sobol (S) and Halton (H) sequences.}}\label{tab:gen_lipschitz_comparison}
	\vspace{-0.2cm}
	\renewcommand{\arraystretch}{1.45}{
		\begin{threeparttable}
		\begin{tabular}{|c|c|c|c|c|c|}
			\hline
			Method	& IA & GS& GS\tnote{$\dagger$} & S & H \\ \hline \hline
			${\gamma}_{l}$	& $32.3514$ & $32.3360$ & $32.3356$ & $32.1041$ & $31.9437$ \\ \hline
			$  \barbelow{\gamma}_{l}$	& $32.2115$ & - & - & - & -\\ \hline
	\end{tabular}
\begin{tablenotes}
	\item[$\dagger$] the Lipschitz constant is computed using \eqref{eq:lip_jacobian}.
\end{tablenotes}
\end{threeparttable}
}
	\vspace{-0.2cm}
\end{table}
\setlength{\floatsep}{10pt}

\vspace{-0.2cm}
\subsection{Lipschitz Parameterization for Synchronous Generator}\label{ssec:gen}
  {The synchronous generator model considered in this example is taken from \cite{Nugroho2019,Nugroho2020DSE}, which is a fourth-order differential equation where the nonlinear dynamics are represented by the function $\m f :\mathbb{R}^4\times \mathbb{R}^4\rightarrow \mathbb{R}^4$ detailed the following form
	\begingroup
	\allowdisplaybreaks
	\vspace{-0.2cm}
	\begin{subequations}{\label{eq:state_space_gen_param}}
		\begin{align}\
			\hspace{-0.0cm}f_1(\m x, \m u) &= -\alpha_1\label{eq:f_1}\\
			\hspace{-0.0cm}f_2(\m x, \m u) &= \alpha_3x_4u_4\cos x_1 - \alpha_3x_3u_4\sin x_1-\alpha_3x_4u_3\sin x_1 \nonumber\\ &\,\quad- \alpha_3x_3u_3\cos x_1 +\alpha_4u_3u_4\cos 2x_1 \nonumber\\ &\,\quad+ \tfrac{1}{2}\alpha_4\left(u_4^2-u_3^2\right)\sin 2x_1 + \alpha_6 \label{eq:f_2}\\
			\hspace{-0.0cm}f_3(\m x, \m u) &= \alpha_8u_4\cos x_1 - \alpha_8u_3\sin x_1\label{eq:f_3}\\
			\hspace{-0.0cm}f_4(\m x, \m u) &= \alpha_{10}u_3\cos x_1 + \alpha_{10}u_4\sin x_1,\label{eq:f_4} 
		\end{align}
	\end{subequations} 
	\endgroup
	where $\alpha_1$, $\alpha_3$, $\alpha_4$, $\alpha_6$, $\alpha_8$, and $\alpha_{10}$ are constants, $\m x$ is the state vector, and $\m u$ is the input vector such that
	\begin{align*}
		&\boldsymbol{x} = \bmat{x_1 \quad x_2 \quad x_3 \quad x_4}^\top = \bmat{\delta \quad \omega \quad e'_{\mathrm{q}} \quad e'_{\mathrm{d}}}^\top \\
		&\boldsymbol{u} = \bmat{u_1 \quad u_2 \quad u_3 \quad u_4}^\top = \bmat{T_{\mathrm{m}} \quad E_{\mathrm{fd}} \quad i_{\mathrm{R}} \quad i_{\mathrm{I}}}^\top,
	\end{align*}
	where 
	$\delta$ is the rotor angle,
	$\omega$ is the rotor speed, and $e'_{\mathrm{q}}$ and $e'_{\mathrm{d}}$ are the transient voltage along $\mathrm{q}$ and $\mathrm{d}$ axes, 
	$T_{\mathrm{m}}$ is the mechanical torque, $E_{\mathrm{fd}}$ is the internal field voltage, $i_{\mathrm{R}}$ and $i_{\mathrm{I}}$ are the real and imaginary currents of the generator \cite{kundur1994power}. }
{In this case we consider a 16-machine, 68-bus system that is extracted from the PST toolbox \cite{chow1992toolbox} by considering Generator 16 in the  network, in which the corresponding parameters are obtained from \cite{chow1992toolbox}.}
The generator's dynamic response are generated from applying a three-phase fault at bus $6$ of line $6-11$ where each generator is using a transient model with IEEE Type DC1 excitation system and
a simplified turbine-governor system \cite{qi2018comparing}.
{
To obtain lower and upper bounds on $\m x$ and $\m u$ (which are needed to construct $\m \Omega$), their maxima and minima are measured from the dynamic response within a $10$ seconds time frame. 
We use the expression from \eqref{eq:gamma_Lipschitz} to compute the corresponding Lipschitz constant. As $\m f(\cdot)$ in \eqref{eq:state_space_gen_param} contains sinusoidal functions, then we implement the $n$-th degree approximations of interval extensions of $\sin(\cdot)$ and $\cos(\cdot)$ (denoted by $\sin^I_n(\cdot)$ and $\cos^I_n(\cdot)$ respectively) from \cite{daumas2009verified}, which are specified as follows
\vspace{-0.01cm}
\begin{align*}
\sin^I_n(\omega) \hspace{-0.05cm}&=\hspace{-0.05cm} \begin{cases}
 [\barbelow{s}(\barbelow{\omega},n),\bar{s}(\bar{\omega},n)], \;\qquad\quad\;\;\text{if}\;\omega\subseteq\left[-\tfrac{\barbelow{\pi}(n)}{2},\tfrac{\barbelow{\pi}(n)}{2}\right] \\
[\barbelow{s}(\bar{\omega},n),\bar{s}(\barbelow{\omega},n)], \;\qquad\quad\;\;\text{if}\;\omega\subseteq\left[-\tfrac{\bar{\pi}(n)}{2},\barbelow{\pi}(n)\right] \\
[\min\left(\barbelow{s}(\barbelow{\omega},n),\barbelow{s}(\bar{\omega},n)\right),1], \,\text{if}\;\omega\subseteq[0,\barbelow{\pi}(n)]  \\
-\sin^I_n(-\omega), \;\qquad\qquad\quad\;\;\text{if}\;\omega\subseteq[-{\barbelow{\pi}(n)},0] \\
[-1,1],\;\qquad\qquad\qquad\quad\;\,\;\text{otherwise}
\end{cases} \\
\cos^I_n(\omega) \hspace{-0.05cm}&=\hspace{-0.05cm} \begin{cases}
[\barbelow{c}(\bar{\omega},n),\bar{c}(\barbelow{\omega},n)], \;\qquad\quad\;\;\text{if}\;\omega\subseteq[0,{\barbelow{\pi}(n)}] \\
\cos^I_n(-\omega),\;\qquad\qquad\qquad\;\text{if}\;\omega\subseteq[-{\barbelow{\pi}(n)},0] \\ 
[\min\left(\barbelow{c}(\barbelow{\omega},n),\barbelow{c}(\bar{\omega},n)\right),1], \,\text{if}\;\omega\subseteq\left[-\tfrac{\barbelow{\pi}(n)}{2},\tfrac{\barbelow{\pi}(n)}{2}\right]  \\
[-1,1],\;\qquad\qquad\qquad\quad\;\,\;\text{otherwise}
\end{cases} 
\end{align*}
where in the above expressions, $\barbelow{s}(x,n)$ and $\bar{s}(x,n)$ are the $n$-th degree partial approximations of $\sin(x)$ whereas $\barbelow{c}(x,n)$ and $\bar{c}(x,n)$ are the $n$-th degree partial approximations of $\cos(x)$ such that the following inequalities hold \cite{daumas2009verified}
\begin{align*}
\barbelow{s}(x,n) &\leq  \sin(x) \leq \bar{s}(x,n),\;\;
\barbelow{c}(x,n) \leq  \cos(x) \leq \bar{c}(x,n).
\end{align*}
Likewise, $\barbelow{\pi}(n)$ and $\bar{\pi}(n)$ are the $n$-th degree partial approximations of the irrational constant $\pi$ such that $\barbelow{\pi}(n)\leq \pi \leq \bar{\pi}(n)$; interested readers can refer to \cite{daumas2009verified} for the detailed mathematical expressions of these approximations. Of course, higher value of $n$ can yield better approximations. In this test, we consider $n = 4$ as it provides sufficient approximations for this model. To obtain sharper bounds when evaluating interval functions, the corresponding subset is subdivided into $10$ segments (or smaller subsets) and the resulting final intervals are computed based on the union of each intervals from the evaluation of these segments---see \cite{ichida1979interval} for more details.

The goals of the first numerical experiment in this section are twofold.  {First, we would like to find the relation between the values of the computed Lipschitz constant and computation time with different precision on Algorithm \ref{alg:BnB-IA} which is carried out by assigning different values for $\epsilon_h$ and $\epsilon_{\Omega}$.} Secondly, we are interested in investigating the practical usefulness of the computed Lipschitz constant for actual observer design. 
To achieve the first objective, two scenarios are considered: first, the value for $\epsilon_h$ is fixed while changing the value of $\epsilon_{\Omega}$ and second we consider the other way, i.e., fixing $\epsilon_{\Omega}$ while using different values for $\epsilon_h$. The results of this experiment are given in Figs. \ref{Fig:gen_1a} and \ref{Fig:gen_1b}. 
It can be seen from Fig. \ref{Fig:gen_1a} that the computed Lipschitz constants $\gamma_l$ are decreasing as the values of $\epsilon_h$ becoming smaller, while on the other hand, the corresponding lower bounds $\barbelow{\gamma}_l$ are increasing.
This indicates that the gaps between $\gamma_l$ and $\barbelow{\gamma}_l$ are getting smaller as $\epsilon_h$ is reduced, in spite of the expense of increasing computation time. Similar trend is not observed when $\epsilon_h$ is fixed while $\epsilon_{\Omega}$ varies---resulting in non-varying Lipschitz values and their lower bounds---see Fig. \ref{Fig:gen_1b}.
These results suggest that, for this particular case, a better Lipschitz constant can be obtained from increasing the algorithm's precision---that is, reducing the value of $\epsilon_h$. 
In comparison with the analytical Lipschitz constant computed using the formula provided in \cite{Nugroho2019} which is equal to $749.8279$, the Lipschitz constant obtained using the proposed approach is indeed far less conservative.
Second, to answer the second objective, we use the computed Lipschitz constant $\gamma_l = 32.3514$ and implement a Lipschitz-based observer for NDS developed in \cite{Phanomchoeng2010} for performing dynamic state estimation on this generator. We find that, with the given Lipschitz constant, the computed observer gain matrix is able to provide a converging estimation error. This finding is corroborated by the converging estimated generator rotor speed $\omega$, denoted by $\hat{x}_2(t)$, to the actual one ${x}_2(t)$, as shown in Fig. \ref{Fig:gen_1c}. 
The results of this numerical test showcase the applicability of the proposed approach in computing Lipschitz constant for practical observer design and state estimation.

In the second numerical test, we again compare the interval-based algorithm (Algorithm \ref{alg:BnB-IA}) with point-based methods including MATLAB's global search and quasi-random search algorithms with Sobol and Halton sequences. The interval-based algorithm is configured with $\epsilon_h = 10$ and $\epsilon_{\Omega} = 10^{-4}$. In addition to using \eqref{eq:gamma_Lipschitz}, we also use the $2$-norm of the Jacobian matrix as provided in \eqref{eq:lip_jacobian} to compute Lipschitz constant with MATLAB's global search.
The results of this test are shown in Tab. \ref{tab:gen_lipschitz_comparison}. It is observed that the interval-based algorithm provides both Lipschitz constant and its lower bound with relatively tight gap, while the Lipschitz constants obtained from  other approaches are comparable. Note that MATLAB's global search with \eqref{eq:lip_jacobian} gives slightly smaller value than from using \eqref{eq:gamma_Lipschitz}, while quasi-random search again gives under-approximations of Lipschitz constant.             
%
}

\vspace{-0.2cm}

\subsection{OSL and QIB Parameterization}\label{ssec:mov_obj}
In the final numerical example, we consider the equation that represents the motion of a moving object in $2-$D plane given in \cite{Abbaszadeh2010} and described in the following state-space equation 
\begin{align}\label{eq:state_space_mov_obj}
\dot{\m x} = \bmat{1&-1\\ 1& 1}\m x + \bmat{-x_1(x_1^2+x_2^2) \\ -x_2(x_1^2+x_2^2)} , \;\; 
\m y  = \bmat{0 & 1}\m x,
\end{align}
with domain of interest $\m \Omega = \left\{\m x\in \mbb{R}^2\,|\, x_i\in [-r,r],\,i=1,2\right\}$ where $r = 5$. Here we aim to \textit{(a)} find the OSL constant $\gamma_s$ and QIB constants $\gamma_{q1}$ and $\gamma_{q2}$ that are \textit{(b)} useful for observer design. To that end, first we focus on finding $\gamma_s$ using the method given in Proposition \ref{prs:gershgorin_max_eigenvalue}.  The corresponding matrix $\m \Xi (\m x)$ for the nonlinear functions in \eqref{eq:state_space_mov_obj} is of the form
\begin{align*}
\m \Xi (\m x) = \bmat{-3x_1^2-x_2^2 & -2x_1x_2 \\ -2x_1x_2 & -x_1^2-3x_1^2}.
\end{align*} 
Since $\m \Xi (\m x)$  in the above is already symmetric, then according to  \eqref{prs:gershgorin_inequality_max_eigenvalue}, $\gamma_s$ can be computed as
\begin{align}
\gamma_s = \max_{i\in\{1,2\}}\Bigg(\max_{\m{x}\in \mathbf{\Omega}}\Bigg(\m \Xi_{(i,i)} +\sum_{j \in\{1,2\}\setminus i}\abs{\m \Xi_{(i,j)}}\Bigg)\hspace{-0.1cm}\Bigg). \label{eq:osl_mov_obj_gershgorin}
\end{align}	
  {In using Algorithm \ref{alg:BnB-IA} to solve \eqref{eq:osl_mov_obj_gershgorin}, we first fix $\epsilon_{\Omega}$ to $10^{-8}$ and use various values for $\epsilon_h$ ranging from $5\times 10^{-4}$ to $10^{-8}$ then in turn, we fix $\epsilon_{h}$ such that $\epsilon_{h} = 10^{-12}$ and use various values for $\epsilon_{\Omega}$.} The results of this numerical example is illustrated in Figs. \ref{Fig:mov_obj_1a} and \ref{Fig:mov_obj_1b}. It can be seen from these figures that, as $\epsilon_h$ getting smaller, the value of $\gamma_s$ is decreasing to a value near zero. Therefore, we conclude that the OSL constant for this system is zero. This results corroborates the fact that, as proven in \cite{Abbaszadeh2010}, the system is indeed globally OSL with $\gamma_s = 0$.   {Notice that in Fig. \ref{Fig:mov_obj_1b}, the computation time is increasing until $\epsilon_{\Omega} = 5\times10^{-6}$ and stays at around $4$ seconds afterwards since for $\epsilon_{\Omega} > 5\times10^{-6}$, the condition $u-l<\epsilon_h$ can be achieved without the need for optimizing the lower bound.} 

Next, we put our attention on determining the QIB constant. To obtain such constant, we solve the following problems using Algorithm \ref{alg:BnB-IA}
	\begin{align*}
	\barbelow{\gamma} &= \min_{i\in\{1,2\}}\Bigg(\min_{\m{x}\in \mathbf{\Omega}}\Bigg(\m \Xi_{(i,i)} -\sum_{j \in\{1,2\}\setminus i}\abs{\m \Xi_{(i,j)}}\Bigg)\hspace{-0.1cm}\Bigg) \\ 
	\gamma_m &= \max_{\m{x}\in \mathbf{\Omega}}\sum_{i\in \{1,2\}}\norm{\nabla_{\hspace{-0.05cm}x} f_i(\m{x})}_2^2, 
	\end{align*}
in which we obtain $\barbelow{\gamma} = -150$ and $\gamma_m = 2.5\times10^4$. According to Proposition \ref{prs:quadratic_inner_bounded}, for any $\epsilon_1,\epsilon_2\in \mbb{R}_{+}$, the QIB constant can be constructed as $\gamma_{q2} = \epsilon_2-\epsilon_1$ and $\gamma_{q1} = \epsilon_1\gamma_s -\epsilon_2\barbelow{\gamma}+\gamma_m$. By setting the constants $\epsilon_1 = 10^5$ and $\epsilon_2 = 10^{-1}$, we obtain $\gamma_{q1} = 25015$ and $\gamma_{q2} = -9999.89$. To test the applicability of the computed OSL and QIB constants for state estimation, we implement an observer developed in \cite{zhang2012full} using the computed constants, in which we are successfully able to get a converging estimation error. Fig. \ref{Fig:mov_obj_1c} depicts the trajectories of the first actual and estimated state for this system.

\vspace{-0.1cm}

\section{Summary and Future Directions}\label{sec:conclusion_future_work}
{
To the best of our knowledge, this work is the first thorough research effort that presents analytical and computational methods to parameterize NDS and various corresponding function sets. Specifically, we propose a novel, unified framework for NDS parameterization for bounded Jacobian, Lipschitz continuous, OSL, QIB, and QB function sets. Our contributions are \textit{(1)} posing and analytically deriving closed-form global maximization problems for NDS parameterization, \textit{(2)} investigating a derivative-free, interval-based algorithm for NDS parameterization, and \textit{(3)} showcasing the proposed methodology for performing parameterization for some NDS models in traffic and power networks. In addition, we  offer various strategies for reducing the complexity in parameterizing NDS using the proposed interval-based global maximization algorithm.  }

{
The presented methods in this paper are not devoid of limitations; several future research directions are worthy of further investigation. First, there exist other classes of nonlinearities that are not considered in this paper such as \textit{incremental quadratic nonlinearity}  \cite{ACIKMESE20111339} and \textit{generalized sector nonlinearity} \cite{Vijayaraghavan2014}.   Parameterizing NDS for these kind of nonlinearities will be studied in our future work. Second, since we do not consider any disturbance---that may be caused by modeling uncertainty or time-varying parameters---on the nonlinearities, then future work will also include investigating robust NDS parameterization.
Third, investigating the impact of obtaining function sets and bounding constants, through our proposed algorithms, on LMI/SDP feasibility for observer and controller designs is another direction worthy of exploring.}

\vspace{-0.1cm}

{
\section*{Acknowledgments}
We would like to thank the associate editor Dr. Leonid Freidovich and the anonymous reviewers for their excellent feedback and constructive comments that improved the content, analysis, and presentation of the paper.}

\vspace{-0.1cm}

\bibliographystyle{IEEEtran}	\bibliography{bib_file}

\appendices

\section{Proof of Theorem \ref{thm:IA_based_algorithm}}\label{appdx:B}
The proof requires the following propositions:
\begin{myprs}\label{thm:global_maximization}
	Let $\m \omega^*$ be a maximizer for problem \eqref{eq:max_basic_problem}. If $C = \{\mathbfcal{S}_i\}^N_{i = 1}$ is a cover in $\mathbf{\Omega}$ that contains a maximizer, then we have $h^*\in[l,u]$ where 
	\begin{align} \label{eq:global_maximization_l_u} 
	l &= \max_{\mathbfcal{S}_i\in C} \inf \left( h^I\hspace{-0.05cm}\left(\mathbfcal{S}_i\right) \right), \quad
	u = \max_{\mathbfcal{S}_i\in C} \sup \left( h^I\hspace{-0.05cm}\left(\mathbfcal{S}_i\right) \right).
	\end{align}   
\end{myprs}
\begin{proof}
	Suppose that the cover $C = \{\mathbfcal{S}_i\}^N_{i = 1}$ in $\mathbf{\Omega}$ contains a maximizer for problem \eqref{eq:max_basic_problem}. As $C$ contains a maximizer, then $\m \omega^*\in \mathbfcal{S}_j$ for some $j\in\mbb{I}(N)$, which in turn implies that $h^* \leq \sup \left(h^I\hspace{-0.05cm}(\mathbfcal{S}_j)\right)$.  Since $\sup \left(h^I\hspace{-0.05cm}(\mathbfcal{S}_j)\right)\leq u$, then it follows that $h^* \leq u$. Next, consider $\mathbfcal{S}_k\in C$ such that $l = \inf \left(h^I\hspace{-0.05cm}(\mathbfcal{S}_k)\right)$. As $\mathbfcal{S}_k\neq 0$, then there exists some points $\m \omega\in \mathbfcal{S}_k$ for which $h(\m \omega)\leq h^*$. However, since $\inf \left(h^I\hspace{-0.05cm}(\mathbfcal{S}_k)\right)\leq h(\m \omega)$, then $l \leq h^*$. This shows that $l\leq h^*\leq u$. 
\end{proof}
\vspace{-0.1cm}
\begin{myprs}\label{thm:removing_subset}
	Let $\m \omega^*$ be a maximizer for problem \eqref{eq:max_basic_problem} and $C = \{\mathbfcal{S}_i\}^N_{i = 1}$ be a cover in $\mathbf{\Omega}$ containing a maximizer. If  there exists $\mathbfcal{S}_j\in C$ for some $j\in\mbb{I}(N)$ such that $\sup \left(h^I\hspace{-0.05cm}(\mathbfcal{S}_j)\right) < l$, then $\m \omega^*\notin \mathbfcal{S}_j$.
\end{myprs}
\begin{proof}
	\vspace{-0.1cm}
	Note that, as the cover $C$ contains a maximizer, then from Proposition \ref{thm:global_maximization} it holds that $l\leq h^*\leq u$. By contradiction, suppose that there exists such $\mathbfcal{S}_j\in C$, where $j\in\mbb{I}(N)$, that $\sup \left(h^I\hspace{-0.05cm}(\mathbfcal{S}_j)\right) < l$ and $\m \omega^*\in \mathbfcal{S}_j$. However, since the assertion $\m \omega^*\in \mathbfcal{S}_j$ implies $h^*\leq \sup \left(h^I\hspace{-0.05cm}(\mathbfcal{S}_j)\right)$, we must have $h^* < l$, which is a contradiction.
\end{proof}
We start the proof by assuming that $\mathbf{\Omega}$ is nonempty such that $\m \omega^*\in \mathbf{\Omega}$. From the initialization, obviously $C$ contains the maximizer. In the first part, we will show that the first loop in Step \ref{alg:first_loop} computes the best upper bound for $h^*$. Note that, due to Proposition \ref{thm:global_maximization}, at every instance we always have $l\leq h^*\leq u$ with $\m \omega^*\in C$.
Since $u = \sup \left( h^I\hspace{-0.05cm}\left(\bar{\mathbfcal{S}}\right) \hspace{-0.00cm}\right)$ where $\bar{\mathbfcal{S}}$ is an element of $C$ with the greatest upper bound, it is divided into two subsets namely $\mathbfcal{S}_l$ and $\mathbfcal{S}_r$. As $\mathbfcal{S}_l\subset \bar{\mathbfcal{S}}$ and $\mathbfcal{S}_r\subset \bar{\mathbfcal{S}}$, by inclusion isotonicity, it then holds that $h^I\hspace{-0.05cm}\left({\mathbfcal{S}}_l\right)\subset h^I\hspace{-0.05cm}\left(\bar{\mathbfcal{S}}\right)$ and $h^I\left({\mathbfcal{S}}_r\right)\subset h^I\hspace{-0.05cm}\left(\bar{\mathbfcal{S}}\right)$, implying that $\sup \left(h^I\hspace{-0.05cm}\left({\mathbfcal{S}_l}\right)\right)<\sup\left( h^I\hspace{-0.05cm}\left(\bar{\mathbfcal{S}}\right)\right)$ and $\sup \left( h^I\hspace{-0.05cm}\left({\mathbfcal{S}_r}\right)\right)< \sup \left(h^I\hspace{-0.05cm}\left(\bar{\mathbfcal{S}}\right)\right)$. This shows that splitting of $\bar{\mathbfcal{S}}$ yields a better upper bound. In Algorithm \ref{alg:opt_bnb}, any $\mathbfcal{S}_i\in C$ satisfying $\sup \left(h^I\hspace{-0.05cm}\left({\mathbfcal{S}_i}\right)\right) < l$ may be discarded since, according to Proposition \ref{thm:removing_subset}, such subset does not contain any maximizer. Thus, we always have $\m \omega^*\in C$ and $l\leq h^*\leq u$. When the first loop in Step \ref{alg:first_loop} terminates, we have $u-l\leq \epsilon_h$---therefore, $h^*$ is $\epsilon_h-${optimal} and the algorithm terminates---or $u-l > \epsilon_h$. In the latter, $\abs{\bar{\mathbfcal{S}}} \leq \epsilon_{\Omega}$ hence it cannot be split further. Either case, we have $u$ as the best upper bound for $h^*$. 

In the second part, it will be shown that $l$ provides the best lower bound for $h^*$. Keep in mind that the temporary cover $\tilde{C}$ may be nonempty only if $u-l > \epsilon_h$ and there exists $\mathbfcal{S}_i\in C$ such that $\abs{{\mathbfcal{S}}_i} > \epsilon_{\Omega}$. Note that $\bar{\mathbfcal{S}}\notin \tilde{C}$. Since we split $\barbelow{\mathbfcal{S}}$ into two subsets namely $\mathbfcal{S}_l$ and $\mathbfcal{S}_r$ where $\barbelow{\mathbfcal{S}}$ is an element of $\tilde{C}$ with the biggest lower bound, then as demonstrated above, we have $h^I\hspace{-0.05cm}\left({\mathbfcal{S}}_l\right)\subset h^I\hspace{-0.05cm}\left(\barbelow{\mathbfcal{S}}\right)$ and $h^I\left({\mathbfcal{S}}_r\right)\subset h^I\hspace{-0.05cm}\left(\barbelow{\mathbfcal{S}}\right)$, implying that $\inf \left(h^I\hspace{-0.05cm}\left(\barbelow{\mathbfcal{S}}\right)\right)<\inf\left( h^I\hspace{-0.05cm}\left({\mathbfcal{S}}_l\right)\right)$ and $\inf \left(h^I\hspace{-0.05cm}\left(\barbelow{\mathbfcal{S}}\right)\right)<\inf\left( h^I\hspace{-0.05cm}\left({\mathbfcal{S}}_r\right)\right)$. This shows that we attain a better lower bound by splitting $\barbelow{\mathbfcal{S}}$. When the second loop in Step \ref{alg:second_loop} terminates, we have obtained a $\epsilon_h-${optimal} solution or every $\mathbfcal{S}_i\in C$ satisfies $\abs{{\mathbfcal{S}}_i} \leq \epsilon_{\Omega}$---thus $\tilde{C} = \emptyset$. At this point, every subsets in $C$ cannot be split further. Thus, when $u-l > \epsilon_h$, $l$ can be updated by evaluating objective function evaluation  of corner points on each $\mathbfcal{S}_i\in C$ and removing $\mathbfcal{S}_j\in C$ satisfying $\sup \left(h^I\hspace{-0.05cm}\left({\mathbfcal{S}_j}\right)\right) < l$. 
At the end of Algorithm \ref{alg:BnB-IA}, we have achieved  $u-l \leq \epsilon_h$ or $u-l > \epsilon_h$---in the former case, we have a $\epsilon_h-${optimal} solution. Either case, $[l,u]$ is the best interval containing $h^*$. Since we know that $\m \omega^*\in C$, there exists $\mathbfcal{S}_j$ for some $j\in\mbb{I}(N)$ for which $\m \omega^*\in \mathbfcal{S}_j$. Thus, it immediately follows that $\m \omega^*\in \prod_{i\in\mbb{I}(N)}\mathbfcal{S}_i$. This ends the proof.
\newqed

\onecolumn

\section{Observer, Controller, and Observer-Based Stabilization for Different Classes of NDS}\label{appdx:A}
Table \ref{tab:ObserverDesign} summarizes several type of observer design procedure via LMI for NDS which nonlinearities satisfy one of the following condition: bounded Jacobian, Lipschitz continuous, one-sided Lipschitz, and quadratically inner-bounded. Table \ref{tab:ControllerDesign} shows some controller and observer-based stabilization design via LMI for NDS with such nonlinearities.

\setlength{\textfloatsep}{10pt}
\begin{table}[h]
	\vspace{-0.2cm}
	\centering 
	\caption{{Observer designs for NDS with different classes of nonlinearity through SDPs or LMIs~\cite{zhang2012full,Phanomchoeng2010,Jin2018,nugroho2018journal}.}}
	\label{tab:ObserverDesign}
	\vspace{-0.2cm}
	\renewcommand{\arraystretch}{2}
	\begin{tabular}{|l| l| l|}
		\hline 
		\textbf{Dynamics and Nonlinearity Class} & \textbf{Observer Structures} & \textbf{Observer Design via LMI} \\[0.5ex] 
		\hline 
		\hline
		$\begin{array}{l}
		\dot{\m x} = \m A\m x  +  \m f(\m x) +\m B \m u\\
		\m y  = \m C \m x +\m D \m u  \\ 
		\textbf{Glob./Loc. Lipschitz:}\\
		\norm{\m f(\m x)-f(\hat{\m x})}_2 \leq \gamma_l\norm{\m x - \hat{\m x}}_2 \\ 	\end{array}$ & $\begin{array} {lcl} 	\dot{\hat{\m x}} = \m A \hat{\m x}+ {{\m f}}(\hat{\m x}) +\m B \m u  \\\;\;\;\;\;\;+\m L(\m y  - \hat{\m y} )\\
		\hat{\m y}  =\m C \hat{\m x}+\m D \m u\\
		\m L=\m P^{-1}\m Y \\
		\text{Variables:}\,\, \m P,\m Y, \kappa 	 \end{array}$ & $\begin{array} {lcl}\mathrm{find}\;\; \m P\succ 0, \m Y, \kappa > 0\\
		\mathrm{s.t.} \;\; \begin{bmatrix}
		\m A ^{\top}\m P + \m P\m A -\m Y\m C \\ - \m C ^{\top}\m Y ^{\top}  +\kappa\gamma_l^2 \m I& * \\
		\m P & -\kappa \m I \end{bmatrix} \prec 0, \end{array}$\\ 
		\hline 
		$\begin{array}{l}
		\dot{\m x} = \m A  \m x  +  \m G\m f(\m x) +\m B \m u \\
		\m y  =  \m C  \m x + \m D  \m u  \\ 
		\textbf{One-Sided Lipschitz, Quadratically} \\
		\textbf{Inner-Bounded:}
		\\ \text{constants} \,\, \gamma_s,\gamma_{q1},\gamma_{q2} \,\, \text{(see Tab.~\ref{tab:nonlinear_class})} 			\end{array}$ & $\begin{array} {lcl} 	\dot{\hat{\m   x}} =   \m A \hat{ \m  x} + \m G{{\m f}}(\hat{\m x}) +\m B \m u \\\;\;\;\;\;\;+ \m L(  \m y  - \hat{ \m  y} )\\
		\hat{ \m  y}  = \m C \hat{  \m x}+ \m D \m   u\\
		\m	L= \dfrac{1}{2}\sigma \bm P^{-1}\m C^\top\\ 
		\text{Variables:}\,\, \m P, \sigma, \epsilon_1, \epsilon_2  \end{array}$ & $\begin{array} {ll} \mathrm{find}\;\;  \m P \succ 0, \{\sigma,\epsilon_1,\epsilon_2\}>0 \\
		\st\;\;	\begin{bmatrix}
		\m 	A^\top \bm P + \bm P\m A + \epsilon_1\gamma_s\m I \\+\epsilon_2\gamma_{q1}\m I-\sigma \m C^\top \m C &* \\ 
		\m G^{\top}\m P+\dfrac{\gamma_{q2}\epsilon_2-\epsilon_1}{2}\m I & -\epsilon_2 \m I  
		\end{bmatrix} \prec 0 
		\end{array}$\\ 
		\hline 
		$\begin{array}{l}
		\dot{\m x} = \m A  \m x + \m B_u \m u + \m G\m f(\m x) + \m B_w \m w  \\
		\m y  =  \m C  \m x +\m D_u \m u + \m D_w  \m w  \\ 
		\textbf{Glob./Loc. Lipschitz:}\\
		\norm{\m f(\m x)-f(\hat{\m x})}_2 \leq \gamma_{l}\norm{\m x - \hat{\m x}}_2 			\end{array}$ & $\begin{array} {lcl} 	\dot{\hat{\m   x}} =   \m A \hat{ \m  x}  + \m G{{\m f}}(\hat{\m x}) +\m B \m u \\ \;\;\;\;\;\;+  \m L(  \m y  - \hat{ \m  y} )\\
		\hat{ \m  y}  = \m C \hat{  \m x} +\m D_u \m u\\
		\m L = \m P^{-1}\m Y\\  
		\text{Variables:}\,\, \m P, \m Y,\kappa,\mu_1,\mu_2\\
		\text{Constant:}\,\, \alpha > 0\\
		\end{array}$ & $\begin{array} {ll} \mathrm{min} \;\; \mu_1 + \mu_2\\
		\st\;\;	\begin{bmatrix} \m A^{\top}\m P+\m P\m A-\m Y\m C& & \\
		-\m C^{\top}\m Y^{\top} + \kappa\gamma_{l}^2\m I + \alpha \m P& * & * \\
		\m G^{\top}\m P & -\kappa \m I & * \\ \m B_w^{\top}\m P-\m D_w^{\top}\m Y^{\top} & \m O & -\alpha\mu_1\m I \end{bmatrix} \prec 0  \\
		\;\;\;\;\;\;\;	\begin{bmatrix} -\m P & * & * \\ \m O & -\mu_2\m I & * \\
		\m Z & \m O & -\m I\end{bmatrix} \prec 0 \\
		\;\;\;\;\;\;\;	 \m P \succ 0, \m Y, \{\kappa,\mu_1,\mu_2\}>0 
		\end{array}$\\ 
		\hline 
		$\begin{array}{l}
		\dot{\m x} = \m A\m x  +  \m G\m f(\m x) +\m B \m u\\
		\m y  = \m C \m x +\m D \m u  \\ 
		\textbf{Bounded Jacobian:}\\
		\barbelow{\m f} \leq \mathrm{D}_x \m f(\m x) \leq \bar{\m f} \\
		\barbelow{b}_{ij}:= \barbelow{f}_{(i,j)},\; \bar{b}_{ij}:= \bar{ f}_{(i,j)}	\end{array}$ & $\begin{array} {lcl} 	\dot{\hat{\m x}} = \m A \hat{\m x} + \m G{{\m f}}(\hat{\m x}) +\m B \m u  \\\;+\m L(\m y  - \hat{\m y} )\\
		\hat{\m y}  =\m C \hat{\m x}+\m D \m u\\
		\m L=\m P^{-1}\m Y \\
		\text{Variables:}\, \m P,\m Y, \m \Lambda 	\\
		\text{Constants:}\,
		\barbelow{c}_{ij} = \frac{1}{2}\left(\barbelow{b}_{ij}+ \bar{b}_{ij}\right), \forall i,j\\
		\bar{c}_{ij} = \frac{1}{2}\left(\barbelow{b}_{ij}- \bar{b}_{ij}\right), \forall i,j\\
		\m W = \m I_g \otimes \m 1_{1\times n_x} \\
		\end{array}$ & $\begin{array} {lcl}\mathrm{find}\;\; \m P\succ 0, \m Y, \m \Lambda\\
		\mathrm{s.t.} \;\; \begin{bmatrix}
		\m A ^{\top}\m P + \m P\m A -\m Y\m C \\ - \m C ^{\top}\m Y ^{\top}  +\m \Theta_1\left(\m{\Lambda}\right)& * \\
		\m W^{\top}	\m G^{\top}\m P + \m\Theta_2\left(\m{\Lambda}\right)& \m\Theta_3\left(\m{\Lambda}\right)\end{bmatrix} \prec 0\\
		\;\;\;\;\;\m\Theta_1\left(\m{\Lambda}\right) = \mathrm{Diag}\left(\left\{\sum_{i=1}^{n}\lambda_{ij}\left(\bar{c}_{ij}^2-\barbelow{c}_{ij}^2\right) \right\}^{n}_{j = 1} \right) \\
		\;\;\;\;\;\m\Theta_2\left(\m{\Lambda}\right)^{\top} = \bmat{\left\{\mathrm{Diag}\left(\left[\lambda_{i1}\barbelow{c}_{i1},\hdots,\lambda_{in_x}\barbelow{c}_{in_x}\right]\right)\right\}^{n}_{i = 1}}\\
		\;\;\;\;\;\m\Theta_3\left(\m{\Lambda}\right) = \mathrm{Diag}\left(\mathrm{vec}\left(\m{\Lambda}\right)\right)\end{array}$\\ 
		\hline 
	\end{tabular}
	\vspace{-0.2cm}
\end{table}
\setlength{\floatsep}{10pt}

\setlength{\textfloatsep}{10pt}
\begin{table}[h]
	\centering 
	\caption{{Controller and observer-based stabilization designs for NDS with different classes of nonlinearity through SDPs or LMIs~\cite{Siljak2002,Yadegar2018,wu2015observer}.}}
	\label{tab:ControllerDesign}
	\vspace{-0.2cm}
	\renewcommand{\arraystretch}{2}
	\begin{tabular}{|l| l| l|}
		\hline 
		\textbf{Dynamics and Nonlinearity Class} & \textbf{Control/Observer Structures} & \textbf{Control Design via LMI} \\[0.5ex] 
		\hline 
		\hline
		$\begin{array}{l}
		\dot{\m x} = \m A\m x  +  \m G\m f(\m x) +\m B \m u\\
		\textbf{Quadratically Bounded:}\\
		\langle \m f(\m x),\m f(\m x)\rangle \leq  \m x^{\top}\m H^{\top}\m H \m x \\ 	\end{array}$ & $\begin{array} {lcl} \m u = \m K \m x\\	
		\m K = \m L\m P^{-1} \\
		\text{Variables:}\,\, \m P,\m L, \gamma	 \end{array}$ & $\begin{array} {lcl}\mathrm{min}\;\; \gamma\\
		\mathrm{s.t.} \;\; \begin{bmatrix}
		\m A\m P + \m P\m A ^{\top} +\m B \m L + \m L^{\top}\m B^{\top} & * &*\\
		\m G^{\top} & -\m I & * \\
		\m H\m P & \m O & -\gamma \m I \end{bmatrix} \prec 0,\\
		\quad\quad \m P \succ 0,\gamma > 0\end{array}$\\ 
		\hline 
		$\begin{array}{l}
		\dot{\m x} = \m A\m x  +  \m f(\m x) +\m B \m u\\
		\textbf{Glob./Loc. Lipschitz:}\\
		\norm{\m f(\m x)-f(\hat{\m x})}_2 \leq \gamma_l\norm{\m x - \hat{\m x}}_2 \\ 	\end{array}$ & $\begin{array} {lcl}  \m u = -\m K \m x\\	
		\m K = \m W\m Q^{-1} \\
		\text{Variables:}\,\, \m Q,\m W\\
		\text{Constant:}\,\, \rho > 0 \end{array}$ & $\begin{array} {lcl}\mathrm{find}\;\; \m Q\succ 0, \m W\\
		\mathrm{s.t.} \;\; \begin{bmatrix}
		\m A\m Q + \m Q\m A ^{\top} -\mB\mW-\mW^{\top}\mB^{\top}+\frac{1}{\rho}\mI & * \\
		\m Q & -\frac{1}{\rho\gamma_l^2} \m I \end{bmatrix} \prec 0, \end{array}$\\ 
		\hline 
		$\begin{array}{l}
		\dot{\m x} = \m A\m x  +  \m f(\m x) +\m B \m u\\
		\m y = \mC \m x \\
		\textbf{One-Sided Lipschitz, Quadratically} \\
		\textbf{Bounded:}
		\\ \text{constants} \,\, \gamma_s,\gamma_{q1},\gamma_{q2} \,\, \text{(see Tab.~\ref{tab:nonlinear_class})}  	\end{array}$ & $\begin{array} {lcl} 
		\dot{\hat{\m   x}} =   \m A \hat{ \m  x} + {{\m f}}(\hat{\m x}) + \m B  \m  u    \\\;\;\;\;\;\;+ \m K(  \m y  - \hat{ \m  y} )\\
		\hat{ \m  y}  = \m C \hat{  \m x}\\ \m u = \m F \hat{\m x}\\	
		\m K = \m R^{-1} \hat{\m K}^{\top}\\ \m F = \hat{\m F}^{\top}\m Q^{-1} \\
		\text{Variables:}\,\, \m Q,\m R, \hat{\m K}, \hat{\m F}, \\
		\quad\quad\quad\quad\;\;\epsilon_1,\epsilon_2,\epsilon_3,\epsilon_4\\
		\text{Constant:}\,\, \phi_1 > 0 \end{array}$ & $\begin{array} {lcl}\mathrm{find}\;\; \m Q\succ 0, \m R\succ 0, \hat{\m K}, \hat{\m F},\epsilon_1>0,\epsilon_2>0,\epsilon_3>0,\epsilon_4>0\\
		\mathrm{s.t.} \;\; \begin{bmatrix}
		\m S_1 & *\\ \m S_2^{\top} & \m S_3\end{bmatrix} \prec 0 \\
		\m S_1 = \bmat{\tilde{\m\Sigma}_{11}(\m Q,\hat{\m F})&*&*&*\\
			\m O & \tilde{\m\Sigma}_{22}(\m R,\hat{\m K},\epsilon_3,\epsilon_4,\gamma_s,\gamma_{q1})&*&*\\
			\mI&\mO&-2\epsilon_2\m I&*\\
			\mO&\mR + (\epsilon_4\gamma_{q2}-\epsilon_3)\mI&\mO&-2\epsilon_4\mI} \\
		\mS_2 = \bmat{-\m B\hat{\m F}^{\top}&\m O&\m Q&\m O\\
			\m O&\m I&\m O&\m O\\ \mO&\mO&\mO&\mI\\ \mO&\mO&\mO&\mO}\\
		\mS_3 = \mathrm{Diag}\left(\bmat{-\frac{\m Q}{\phi_1}&-\phi_1\m Q&-\psi(\epsilon_1,\epsilon_2)\mI& -\zeta(\epsilon_1,\epsilon_2)\mI}\right)
		\end{array}$\\ 
		\hline  
	\end{tabular}
	\vspace{-0.3cm}
\end{table}
\setlength{\floatsep}{10pt}

\end{document}

%% file: preamble.tex
\usepackage{epsfig,color,amsmath,cite}
\usepackage{amsthm} 
\usepackage{amsmath}    
\IEEEoverridecommandlockouts
\usepackage{bm}
\usepackage{epstopdf}
\usepackage{amssymb}
\usepackage{url}
\usepackage{enumitem} 
\usepackage{multirow}
\usepackage{hhline}
\usepackage{booktabs}
\usepackage[linesnumbered,boxed,commentsnumbered,ruled,vlined,longend]{algorithm2e}
\usepackage{comment}
\newcommand{\eps}{\varepsilon}

\DeclareMathOperator*{\subjectto}{subject\ to}

\DeclareMathOperator*{\argmax}{arg\ max}

\makeatother
\DeclareMathAlphabet\mathbfcal{OMS}{cmsy}{b}{n}

\newtheorem{theorem}{Theorem}
\newtheorem{mydef}{Definition}
\newtheorem{mylem}{Lemma}

\newtheorem{myrem}{Remark}
\newtheorem{asmp}{Assumption}
\newtheorem{mycor}{Corollary}
\newtheorem{myprs}{Proposition}
\newtheorem{exmpl}{Example}

\newcommand{\reducewordspace}[0]{\fontdimen2\font=0.645ex}



\usepackage{stackengine}
\newcommand\barbelow[1]{\stackunder[1.2pt]{$#1$}{\rule{.8ex}{.075ex}}}

\newcommand{\mat}[1]{\boldsymbol{#1}}

\newcommand{\bmat}[1]{\begin{bmatrix} #1 \end{bmatrix}}

\providecommand{\mA}{\ensuremath{\mat{A}}}
\providecommand{\mB}{\ensuremath{\mat{B}}}
\providecommand{\mC}{\ensuremath{\mat{C}}}

\providecommand{\mG}{\ensuremath{\mat{G}}}

\providecommand{\mI}{\ensuremath{\mat{I}}}

\providecommand{\mO}{\ensuremath{\mat{O}}}

\providecommand{\mR}{\ensuremath{\mat{R}}}
\providecommand{\mS}{\ensuremath{\mat{S}}}

\providecommand{\mW}{\ensuremath{\mat{W}}}

\newcommand{\st}{{\rm s.t.}}

\newcommand{\m}{\boldsymbol}
\allowdisplaybreaks[4]
\usepackage[colorlinks = true,
linkcolor = blue,
urlcolor  = blue,
citecolor = blue,
anchorcolor = blue]{hyperref}

\newcommand{\mc}[1]{\mathcal{#1}}
\newcommand{\mbb}[1]{\mathbb{#1}}
\newcommand{\mr}[1]{\mathrm{#1}}
\usepackage[framemethod=TikZ]{mdframed}
\mdfdefinestyle{MyFrame}{%
	linecolor=black,
	outerlinewidth=1.25pt,
	roundcorner=1.25pt,
	innerrightmargin=5pt,
	innerleftmargin=5pt,}
	
\usepackage[noabbrev]{cleveref}

\usepackage{mathtools}

\DeclarePairedDelimiter\abs{\lvert}{\rvert}%
\DeclarePairedDelimiter\norm{\lVert}{\rVert}%

\makeatletter
\let\oldabs\abs
\def\abs{\@ifstar{\oldabs}{\oldabs*}}
\let\oldnorm\norm
\def\norm{\@ifstar{\oldnorm}{\oldnorm*}}
\makeatother

\usepackage[english]{babel}
\usepackage[utf8]{inputenc}
\usepackage[super]{nth}

\usepackage{graphicx}
\usepackage{float}
\usepackage[caption = false]{subfig}

\usepackage{array}
\usepackage{threeparttable}

\usepackage[english]{babel}
\usepackage[utf8]{inputenc}
\usepackage[super]{nth}

\usepackage{pifont}
\newcommand{\cmark}{\ding{51}}%

\def\newqed{{\null\nobreak\hfill\color{black}\ensuremath{\blacksquare}}}